%% file: ecai_main.tex
\documentclass{ecai} 
\usepackage{graphicx}
\usepackage{soul}
\usepackage{algorithm}
\usepackage{algorithmicx}
\usepackage[noend]{algpseudocode}
\usepackage{amsfonts}
\usepackage{natbib}
\usepackage{enumitem}

\usepackage{etoolbox}
\newtoggle{arxiv}  %% Setting toggle arxiv to true generates arxiv version.
\toggletrue{arxiv}

\usepackage{latexsym}
\usepackage{amssymb}
\usepackage{amsmath}
\usepackage{amsthm}
\usepackage{booktabs}
\usepackage{enumitem}
\usepackage{graphicx}
\usepackage{color}

\newtheorem{lemma}{Lemma}
\newtheorem{claim}{Claim}
\newtheorem{definition}{Definition}

\setlist[itemize]{leftmargin=9pt,labelsep=4pt,noitemsep,topsep=3pt}
\setlist[enumerate]{leftmargin=10pt,labelsep=10pt,noitemsep,topsep=3pt}
\setlist[description]{leftmargin=5pt,labelsep=2pt,noitemsep,topsep=3pt}

\newtheorem{theorem}{Theorem}

\newcommand{\Ut}{T}
\newcommand{\Si}{S}

\newenvironment{proofof}[1][\proofname]{
  \proof[\bfseries Proof of #1]
}{
  \endproof
}

\newenvironment{proofideaof}[1][\proofname]{
  \proof[\bfseries Proof idea of #1]
}{
  \endproof
}

\newenvironment{statementof}[1][\proofname]{
  \proof[\bfseries Statement of #1]
}{
}

\newcommand{\Einput}{E^+}

\newcommand{\agset}{X} %% Set of agents.
\newcommand{\rset}{Y} %% Set of resources.

\newcommand{\BibTeX}{B\kern-.05em{\sc i\kern-.025em b}\kern-.08em\TeX}

\begin{document}

\begin{frontmatter}

%%% Use this command to specify your submission number.
%%% In doubleblind mode, it will be printed on the first page.

\paperid{9712} 

%%% Use this command to specify the title of your paper.

%%% Ravi changed the following title (to the
%%% one mentioned in the abstract submission to
%%% ECAI.

%\title{Guidelines for Preparing a Paper for the \\
%European Conference on Artificial Intelligence}

\title{Facilitating Matches on Allocation Platforms}

\author[A]{\fnms{Yohai}~\snm{Trabelsi}\thanks{Corresponding Author. Email: yohai.trabelsi@gmail.com.} \footnote{Work done while the author was affiliated with Bar-Ilan University, Israel. }}

\author[B]{\fnms{Abhijin}~\snm{Adiga}}
\author[C]{\fnms{Yonatan}~\snm{Aumann}}
\author[C]{\fnms{Sarit}~\snm{Kraus}}
\author[B]{\fnms{S. S.}~\snm{Ravi}}

\address[A]{Harvard University}
\address[B]{Biocomplexity Institute, University of Virginia}
\address[C]{Dept. of Computer Science, Bar-Ilan University}

%%%%%%%%%%%%%%%%%%%%%%%%%%%%%%%%%%%%%%%%%%%%%%%%%%%%%%%%%%%%%%%%%%%%%%%%
\input{abstract}

\end{frontmatter}

\input{introduction}
\input{related}
\input{definitions}

\input{algo}

\input{many}
\input{experiments}
\input{conclusion}

\clearpage 

\noindent
\textbf{Acknowledgment:}~ We thank the ECAI-2025 reviewers for carefully reading the paper and providing very helpful comments.
*The research of Yonatan Aumann is supported in part by ISF
grant 3007/24. The research of Sarit Kraus is supported in part by
ISF grant 2544/24.

\medskip 

\bibliography{refs}

\clearpage
\input{appendix}

\end{document}

%% file: abstract.tex
\begin{abstract}

We consider a setting where goods are allocated to agents by way of an allocation platform (e.g., a matching platform). An ``allocation facilitator'' aims to increase the overall utility/social-good of the allocation by encouraging (some of the) agents to relax (some of) their restrictions. At the same time, the advice must not hurt agents who would otherwise be better off.
Additionally, the facilitator may be constrained by a ``bound'' (a.k.a. `budget'), limiting the number and/or type of restrictions it may seek to relax. We consider the facilitator's optimization problem of choosing an optimal set of restrictions to request to relax under the aforementioned constraints.  
Our contributions are three-fold: 
(i) We provide a formal definition of the problem, including the participation guarantees to which the facilitator should adhere.  We define a hierarchy of participation guarantees and also consider several social-good functions.    
(ii) We provide polynomial algorithms for solving various versions of the associated optimization problems, including one-to-one and many-to-one allocation settings.
(iii) We demonstrate the benefits of such facilitation and relaxation, and the implications of the different participation guarantees, using extensive experimentation on three real-world datasets.

\end{abstract}

%% file: introduction.tex
\section{Introduction}
Recently, allocation-platforms/matching-platforms, which allocate  resources of one sort or another to users, are being deployed for a variety of applications in both the public and private sectors, including in welfare and social services~\citep{huang2023methodology,pan2023matching,aguma2022matching}. Some examples are allocating home healthcare demand with service providers \citep{lin2021matching}, on-demand housekeeping platforms \citep{yu2022pricing}, 
government platforms for providing housing assistance to homeless individuals \citep{sharam2018matching}, ride-sharing platforms \citep{bao2023mathematical}, 
sharing parking spaces \citep{tang2023efficient}, and volunteer matching platforms \citep{slingerland2018empowering}.  In these platforms, users specify their resource requirements and constraints, and the platforms aim to optimize the allocation of resources to users 
where resources can be shared among several users.

The key stakeholders in the resultant allocation are clearly the users. In many cases, however, there may be additional stakeholders.  For example, local welfare authorities are justifiably interested in increasing the number of homeless individuals awarded housing assistance (see e.g., \cite{chan2018utilizing}), and university administration is interested in maximizing the number of courses for which classrooms have been successfully allocated (see e.g., \cite{frimpong2015allocation}).
At times, these stakeholders may be able to directly determine, or make changes to the allocation, but more often than not, the allocation procedure itself is fixed - for regulatory, commercial, or technical reasons (see e.g., ~\cite{trabelsi2022resource}).  
For example, New York City offers affordable housing opportunities through the {\it Housing Connect portal}~\cite{nyc_housing_connect}, which serves as an allocation and matching platform. The algorithm used for this allocation is governed by multiple laws and regulations, including: (i) the federal Fair Housing Act, (ii) the NYC Human Rights Law, (iii)NYC HPD regulations, and (iv) NYC HDC rules. 
These regulations determine the allocation, which does not change. 
At the same time, the NYC Department of Housing Preservation and Development (HPD) established an advisory initiative, called the {\it Housing Ambassadors Program}~\cite{hpd_housing_ambassadors}, to help people navigate and use the allocation platform effectively. In their website, they emphasize that the Housing Ambassadors do not provide housing directly, and they cannot guarantee that an applicant will receive an affordable unit through the lottery.

In such cases, interested parties -- which from now on we term \emph{facilitators} -- can still shape the resulting allocation by assisting and advising users in selecting the priorities and constraints they enter into the allocation platform.   It is important to stress that such advice need not be viewed as a form of manipulation, neither of the platform nor of the users. Indeed, users frequently do not know how best to express their true constraints, and such interventions -- if done right -- can benefit all~\cite{Slaugh2016}. As such, we only consider \emph{impartial} facilitators whose priorities are aligned with the overall social good, without any preference for one user or another.  In this paper, we study advice provisioning with such impartial facilitators: what advice should they provide? What guarantees should/can be offered to the users, both those following the advice and those who do not? 

In this paper, we assume that with appropriate guarantees, agents are willing to accept the advice. Providing them knowledge about whether they have been, or are currently, guaranteed a resource— either through the facilitator to ease constraints or the allocator as an act of transparency— further enhances their receptiveness.

For concreteness, we consider the following stylized model (see
Figure~\ref{fig:examplefc} for an example). 
There is an allocation platform that, given (i) a set $X$ of agents - each agent $x_i$ with \emph{demand level} $d_i$ for the number of resources it needs, (ii) a set $Y$ of resources, and (iii) a binary compatibility relation $E$ between agents and resources, outputs a maximum allocation of resources 
to agents~(where the maximum is in the sense of resource utilization). 
The compatibility relation $E$ represents the compatibility \emph{as provided to the platform by the agents}. The facilitator can advise agents to add additional compatibilities $\hat{E}$, thereby enlarging the input compatibility relation to $\Einput= E\cup \hat{E}$.
If offered no advice, then $\Einput=E$.
%\footnote{We do not allow the facilitator to persuade an agent to remove compatibilities, as this may be considered manipulative.} 
Adding any such new compatibility $e\in \hat{E}$ is associated with some \emph{discomfort level} $\rho(e)$. 
We consider a class of problems where the facilitator's
objective is to maximize resource allocation with the constraint
that the \emph{aggregate cost} does not exceed a specified bound. Here, the aggregate cost
is some function of all~$\rho(e)$s, such as sum of 
all~$\rho(e)$s or the count of strictly positive $\rho(e)$ values (which is the number of proposed relaxations).
Figure~\ref{fig:examplefc} shows solutions for different 
functions. 
The incompatibility 
between a course and a classroom in Figure~\ref{fig:examplefc} can arise due to various factors
such as seating capacity, commute distance, and accessibility issues.
Optimal solutions for various
scenarios (bound type and value---participation guarantee)
 are shown, all of which can be achieved by our framework.
 Expansions of the acronyms used in these
 scenarios are given under ``Summary of Contributions''.
 Formal definitions of the guarantees are provided in
 Section~\ref{sec:definitions}.
 Note that even for such a simple example, the solutions can be very different
for different scenarios.

We propose the following two requirements to ensure that the proposed relaxation does not harm the agents while encouraging cooperation among them:

\begin{description}
    \item{(i)} \emph{No agent is harmed:} Any agent that was guaranteed to be granted an allocation prior to the facilitator's actions, is also guaranteed so following the facilitator's actions.
    \item{(ii)} \emph{Participating agents benefit:} Any agent that is asked by the facilitator to add a compatibility (and does so) is guaranteed to be granted an allocation. 
\end{description}
\begin{figure}
\centering
\includegraphics[width=\columnwidth]{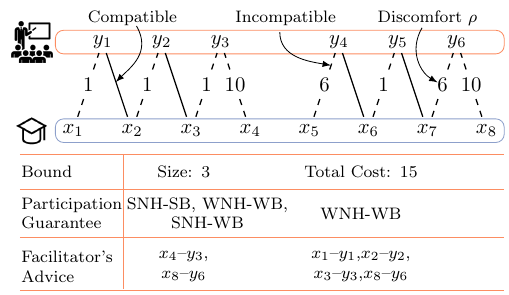}
\caption{An example of a course--classroom allocation platform where every
course needs to be matched to one classroom. Solid edges represent a compatible pair of a course and a classroom, dashed edges indicate an incompatible pair with a finite relaxation cost, and absence of an edge indicates an infinite relaxation cost incompatibility.
An example that distinguishes between the different guarantees is in Figure~\ref{fig:examplepcfc}.
}
\label{fig:examplefc}
\end{figure}
It should be noted that the facilitator cannot secure these guarantees by directly determining the allocation, as the platform is not under its control.
%(e.g., the domain is huge and making hand crafter changes might affect other agents).  
Rather, the facilitator must secure these guarantees by properly designing its advice.

Given such a setting, the facilitator’s objective is to devise a set of incompatibilities that it will request agents to relax, which (i)~maximizes the resultant allocation, while (ii)~maintaining the overall cost within the given bound, and (iii)~maintaining the no-harm and the participating agents guarantees.
Agents are incentivized to follow the facilitator's recommendation as this guarantees them an allocation.
%Throughout this paper, we assume that the agents are sufficiently risk-averse (\cite{pratt1978risk}). Under this assumption, some agents might prefer to follow the facilitator’s advice to mitigate the risk of not being allocated a resource.

\paragraph{Summary of Contributions.}
~\newline
\noindent 
{\underline{(a) Problem formulation:}}~
We formally define the problem, including 
\begin{itemize}
\item Participation guarantees: We consider two forms of ``no harm'' participation guarantees: \emph{strong no harm guarantee (SNH)}, and \emph{weak no harm guarantee (WNH)}, and two forms of benefit to relaxers guarantees: \emph{strong benefit guarantee (SB)}, and \emph{weak benefit guarantee (WB)}. In both cases, the strong guarantee holds regardless of the number of agents adhering to the facilitator's advice, while the weak guarantee holds if all agents comply with the facilitator's recommendations. Using the defined guarantees, we define three guarantee combinations: SNH-SB, WNH-WB, and SNH-WB. The WNH-SB combination is omitted because, in practice, ensuring a strong guarantee for no harm is more important than for the benefit to relaxers.    

\item Aggregated cost:
Inspired by Faliszewski and Rothe~\cite{Faliszewski_Rothe_2016}, we consider two functions for aggregating discomfort from relaxing individual incompatibilities: size  (total relaxations) and total cost (sum of discomforts).

\item 
We consider three allocation settings: (i)~one-to-one where a
resource is allocated to at most one agent and an agent is 
assigned at most one resource, (ii)~many-to-one 
where a resource is allocated to at most  one agent but an agent 
could be assigned multiple resources, and (iii)~one-to-many where an agent is assigned at most one resource but a resource could be allocated to multiple agents (e.g., by sharing a classroom).
\end{itemize}

\noindent
\ul{(b) Polynomial algorithms:}
We provide polynomial algorithms to solve each of the problem variants, for the three participation guarantees combinations, and the two aggregation functions.
For the reason mentioned above, we omit the combination of weak no harm and strong benefit.
%(We omit the combination of weak no harm and strong %benefit, as, in practice, no harm is more crucial than %the benefit to relaxers.).
%We note that while some of the algorithms are relatively simple, the bulk of the work is in proving their correctness (particularly for the case of partial-compliance guarantee).  

\noindent
\ul{(c) Experimental study:}~ 
We applied the devised algorithms to three real-world datasets, and conducted experiments in all three allocations settings (one-to-one, one-to-many and many-to-one).  
In each, we study the improvement in allocation sizes obtained by the facilitator under the different participation guarantees.  We show that for all guarantees, a significant increase in allocation size is obtained. Comparing the performance of the different participation guarantees, we show that if all agents comply with the facilitator's advice, then the stronger guarantees result in somewhat smaller allocations than the weaker guarantees, but using the stronger guarantees is more robust to agents' failure to comply with the advice (as one would expect in practice).

%\paragraph{The Appendix.} Due to space limitations, the proofs of all lemmata, as well as some experimental analysis are deferred to the appendix.

%% \paragraph{Paper Structure.}  Immediately hereunder, we review some of the related work.  Section \ref{sec:definitions} provides the formal definitions and describes the algorithmic template employed. Our solutions for FCFair MCFair and PCFair guarantees are described and proved in Sections~\ref{sec:pc}, \ref{sec:fc} and \ref{sec:mc} respectively. Due to lack of space, the proofs of most lemmata are deferred to the Appendix. 
%% %\textcolor{red}{(I changed ``all lemmata'' to ``most %lemmata''
%% %since the current version includes a proof of Lemma~3. -- %Ravi)} \ycomment{fine}
%% The experimental results are presented in Section~\ref{sec:experiments}, and Section~\ref{sec:concl} concludes with a brief summary.
%\aacomment{Removed paper stucture paragraph.}

%% file: related.tex
\section{Related Work}
\label{sec:related}

\noindent
\textbf{Resource allocation in multi-agent systems.}
Many references discuss
how agents express their requirements, identify efficiently solvable allocation problems  
and provide methods for evaluating the
corresponding algorithms
(see e.g.,
\citep{Chevaleyre-etal-2006,Gorodetski-etal-2003,Dolgov-etal-2006}). 
Nguyen et al.~\cite{nguyen2013survey}
provide a good
survey on the complexity and approximability of
problems in this area. 
Zahedi et al.~\cite{zahedi-etal-2020,zahedi2023didn} 
present a 
method where dissatisfied agents 
can challenge the proposed allocation using counterfactual queries.
Relaxing the criteria for compatibility in kidney matching
is studied in \citep{rao2009comprehensive,kilambi2023evaluation}. However, their focus is on general criteria, not on specific agents (patients). 
Methods for active advice generation for a single agent appear in
\citet{trabelsi2022resource}. A multi-round setting where several
dissatisfied agents are given advice is presented in
\citet{trabelsi2023resource}. While these papers focus on agent satisfaction, our work additionally emphasizes the role of the facilitator.

\noindent
\textbf{Participation in all maximum matchings.}
Costa~\cite{costa1994persistency} presented an algorithm to partition a graph's edge set into three subsets: edges that participate in all maximum matchings, in some maximum matchings, and in none of the maximum matchings. 
Irving et al.~\cite{irving2006rank} 
showed how to efficiently compute the Dulmage–Mendelsohn decomposition~\cite{dulmage1958coverings} of a bipartite graph, enabling direct identification of the set of nodes participating in all maximum matchings. Zhang et al.~\cite{zhang2017efficient} showed another algorithm for this task. However, these works do not consider how changing the graph edges affects this set.

\noindent
\textbf{Modifying the graph for improving the allocation.}
\citet{boehmer2021bribery} considered bribery and external manipulations for providing participation guarantees to an agent pair in stable matchings. \citet{chen2023optimal,gokhale2024capacity} and \citet{bobbio2023capacity} considered adjusting the resource capacities for having a many-to-one matching with some desired properties (e.g., a stable matching). 
However, these works do not consider the criteria of participation in all possible maximum allocations. 
Participation in all maximum allocations is preferred in our context, as the allocator can choose any maximum allocation and is not restricted to any other property (like stability).

%% file: definitions.tex
\section{Definitions and Problem Formulation}
\label{sec:definitions}
Here, we provide definitions for the one-to-one setting. The extension to the many-to-one case is provided in Section~\ref{sec:mw}.

\subsection{Preliminaries}
\label{sec:prelims}
\paragraph{The setting.} The setting consists of sets~\agset~ of agents,~\rset~ of resources, and~$E\subseteq X\times Y$ of \emph{compatible pairs}.  Here, $(x,y)\in E$ means that agent $x$ is willing to be allocated the resource $y$  (without relaxing her preferences). 
Technically, the triple $G=(X,Y,E)$ is simply a 
bipartite graph. 

\paragraph{Discomfort.} For incompatible pairs $(x,y)\not\in E$, there is a \emph{discomfort function} 
$\rho:\overline{E}\rightarrow \mathbb{R}^+\cup \{\infty\}$, where $\rho(x,y)$ reflects the ``discomfort'' for agent $x$ of relaxing the incompatibility $(x,y)$,  
that is, the discomfort that would be experienced by $x$ if allocated to $y$.  If $\rho(x,y)=\infty$ then $y$ cannot be allocated to $x$, and $(x,y)$ are deemed \emph{totally incompatible}. We let $E_R=\{(x,y) | (x,y)\not\in E, \rho(x,y)<\infty\}$ denote the set of \emph{relaxable} incompatibilities. 

\paragraph{Relaxations and cost bound.}  
Given $(X,Y,E)$ and $\rho$, we seek to relax some of the incompatibilities in order to increase the size of the resulting maximal allocation (produced by the allocation platform). Totally incompatible pairs cannot be relaxed. 
Thus, technically, a relaxation is a set $\hat{F}\subseteq E_R$. We denote
by $X(\hat{F})$ the set of agents that participate in the relaxation
$\hat{F}$. The \emph{aggregate cost} of this relaxation is obtained by
aggregating discomfort induced by $\hat{F}$. We consider two \emph{aggregation functions}:
\begin{itemize}
\item Total Cost: $\Ut(\hat{F}):=\sum_{(x,y)\in \hat{F}}\rho(x,y)$. 
\item Size: $\Si(\hat{F}):= |\hat{F}|$.

\end{itemize}
We assume that the facilitator has a \emph{cost bound}~$\beta$  
on the aggregate cost of the relaxation it chooses.

\paragraph{Allocations.} 
In the one-to-one case, an allocation $M\subseteq F$ is a matching and a maximum allocation is a maximum cardinality matching in the graph $G=(X,Y,F)$. An allocation in the many-to-many, many-to-one, and one-to-many cases are defined similarly.
Given a set $F\subseteq E$, we denote by $\mu(F)$ the size of a maximum 
allocation of $F$,
and by $\Gamma(F)\subseteq X$  the set of agents that participate in all maximum allocations of $F$. If several such allocations exist, one of the solutions is picked
arbitrarily.

\paragraph{Minimal relaxation.} A relaxation $\hat{F}$ is minimal 
if~$\mu(E\cup({\hat{F}\setminus \{e\}}))
< \mu(E\cup{\hat{F}})$, for any $e\in \hat{F}$.

\subsection{Participation Guarantees}
As explained in the introduction, the relaxation advice provided by the facilitator must provide guarantees both to the agents participating in the relaxation and to those not. The followings four guarantees are used for defining three guarantee combinations:
\begin{description}
        \item {\underline{Strong No Harm:}} 
        Any agent that participates in all maximum matchings prior to any relaxation continues to have this benefit after the relaxation:
        $\Gamma(E)\subseteq \Gamma(E\cup \hat{F})$, $\forall \hat{F}\subseteq \hat{E}$. 
        \item {\underline{Strong Benefit to relaxers:}} All relaxing agents are guaranteed to participate in any maximum matching: $X(\hat{F})\subseteq \Gamma(E\cup \hat{F})$, $\forall \hat{F}\subseteq \hat{E}$.
        \item {\underline{Weak No Harm:}}  $\Gamma(E)\subseteq \Gamma(E\cup \hat{E})$ 
        \item {\underline{Weak Benefit to relaxers:}} 
        $X(\hat{E})\subseteq \Gamma(E\cup\hat{E})$         
    \end{description}
    
Based on the above, we define the following three guarantee combinations:
\begin{description}
\item {\underline{Strong No Harm, and Strong Benefit to relaxers (SNH-SB):}}
Participation is guaranteed even if some agents do not follow the facilitator's advice. Using this guarantee is preferred when it is expected that several agents will not comply.
\item {\underline{Weak No Harm, and Weak Benefit to relaxers  (WNH-WB):}}
participation is guaranteed, but only under the assumption that all agents follow the facilitator's advice. 
This guarantee is preferred when all agents are known to comply as it often results in a larger allocation.
\item {\underline{Strong No Harm, and Weak Benefit to relaxers (SNH-WB):}}
the No-harm guarantee holds even with partial compliance, but benefit to relaxers is only guaranteed if all agents comply. In the intermediate case, where we expect very high compliance, the facilitator should weigh the costs of losing the Strong No Harm and/or the strong Benefit to relaxers guarantees and the dissatisfaction of some agents against the benefits of increasing the allocation size and choose one of the three guarantees (See Section~\ref{sec:experiments} for details).
\end{description}

It is easy to see that SNH-SB $\Rightarrow$ SNH-WB $\Rightarrow$ WNH-WB.  The  example in Figure~\ref{fig:examplepcfc} 
shows that the hierarchy is strict.

\begin{figure}[htb]
\centering

\includegraphics[width=0.9\columnwidth]{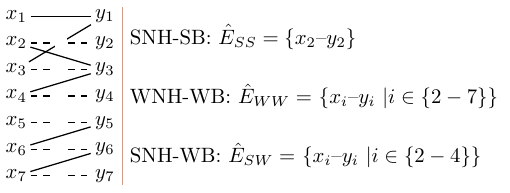}

\caption{$\hat{E}_{SS}$, $\hat{E}_{WW}$ and $\hat{E}_{SW}$ are relaxation sets for SNH-SB, WNH-WB and SNH-WB.  Strictness follows from the fact that $\hat{E}_{SS}\subset \hat{E}_{SW} \subset \hat{E}_{WW}$. 
\label{fig:examplepcfc}}
\end{figure}

\iffalse 
%%%%%%%%%%%
\textcolor{red}{(The strictness may not be clear to the readers
from the figure. It may be useful to include a short discussion
on the strictness in the supplement. -- Ravi)}
\ycomment{Added clarification to the caption of the figure. Ravi: is it sufficient?}
%%%%%%%%%
\fi 
%\textcolor{red}{(Please consider adding a red 
%horizontal line after the last line of the
%table in Figure~2. -- Ravi)}
%\ycomment{Ravi: Is it better?}

\subsection{Problem Statement}
Given the above definitions, the facilitator's optimization problem is the following:
 %Actually, the two different match guarantees and the three different minimization objectives define six different core problems, and 12 different combined objective problems, as follows:
\begin{definition}
    \label{def:prob}
    Given $X,Y,E,\rho$, a participation guarantee $\mbox{PG}\in \{
    \mbox{WNH-WB},\mbox{SNH-WB},\mbox{SNH-SB}\}$, aggregated cost bound $\beta$, and aggregation function $g\in \{\Ut(\cdot),\Si(\cdot)\}$, find a relaxation $\hat{E}$ that is PG,  $g(\hat{E})\leq \beta$, and $\mu(E\cup \hat{E})$ is the maximum among all such relaxations.
\end{definition}

All in all, with three possible participation guarantees, and two aggregation functions, the above defines 
six 
%\textcolor{red}{(Should ``nine'' be changed to
%``six'' here? -- Ravi)}
optimization problems.
We use the term ``solution'' to such an optimization problem
to refer to a set of relaxations.

\subsection{An example}
In Figure~\ref{fig:examplefc} we match eight courses with six available classrooms. 
Initially, the course $x_2$ is compatible with the classroom $y_1$, $x_3$ with $y_2$, $x_6$ with $y_4$ and $x_7$ with $y_5$. 
The discomfort levels can vary widely; low for minor inconveniences 
such as the classroom being far from the offices of the tutor 
($x_1$ and~$y_1$) to  major restrictions such as inadequate capacity, in
which case, the tutor might have to make necessary arrangements to 
accommodate extra students ($x_4$ and~$y_3$).
There can be incompatibilities with infinite discomfort. For example, 
students might require hearing accessories, which might be absent in 
the classroom ($x_3$ and~$y_1$).
%%the cost for making the room $x_1$ compatible to the resource $y_1$ is low since $y_1$ has the appropriate capacity for $x_1$ and it's only problem is that it is not very close to the offices of $x_1$'s instructor. Allowing $x_5$ and $y_4$ to be matched costs more, since the classroom $y_4$ is significanly bigger than the size needed by $x_3$ and the instructor of $x_3$ suffers from teaching in too big classrooms. The edge between $x_4$ and $y_3$ costs $10$ since the seats in $y_3$ are less than required for $x_4$ and in this case some of the students of $x_4$ might need to add chairs to the classroom and make it too crowded. Finally $x_3$ will not be matched to $y_1$ at any cost since one of $y_1$'s students needs hearing accessories and these are not available in $y_1$. In this case we say that the cost for relaxing the edge of $x_3$ and $y_1$ is $\infty$.
Figure~\ref{fig:examplefc} demonstrates the difference between the bound types and the different guarantees.
Figure~\ref{fig:examplepcfc} highlights the differences between various participation guarantees.
%\aacomment{Needs to be resolved}

%% file: algo.tex
\section{Algorithmic Solutions}
\label{sec:algorithms}

Here we introduce efficient algorithms for the different guarantees. In Section~\ref{sec:pc}, we present an algorithm for the strong no harm strong benefit guarantee (SNH-SB). In Section~\ref{sec:mc}, we demonstrate how to adapt this algorithm for the strong no harm weak benefit guarantee (SNH-WB). 
\iftoggle{arxiv}
{For space reasons, the proofs for the SNH-WB algorithm, as well as the algorithmic solution for WNH-WB, are deferred to Section~\ref{sec:appalg} of the appendix.}
{For space reasons, some proofs for the SNH-WB algorithm, as well as the algorithmic solution for WNH-WB, are deferred to the full version of the paper~\cite{Trabelsi-etal-ecai-2025}.}

\subsection{Strong no harm strong benefit to relaxers (SNH-SB)}
\label{sec:pc}

We first consider the stronger (and more complex) Strong no harm strong benefit to relaxers guarantee  (SNH-SB). We provide an efficient algorithm 
(Algorithm \ref{alg:core}) for the problem and prove its correctness.
The core of the algorithm is the weight assignment and the maximum weighted matching calculation~(lines 6-8). 
By assigning appropriate weights to the edges, we ensure that the computed relaxation will be SNH-SB and will have the maximum number of allocations among SNH-SB pairs of the graph $G_k$. 
On the loop in lines~3-13, the algorithm adds $k$ dummy agents and connects them to all resources of $Y$. The dummy agents take the places of other agents in the weighted matching and this way bound the total cost for the non-dummy edges. 
The algorithm picks the smallest $k$ for which the total cost of a maximum weighted matching (without the dummy edges) does not exceed the bound and removes the added agents for constructing a solution. 
We will show that the algorithm returns an optimal solution.

%Given a graph, $G=(X,Y,E\cup E_R)$, for each number, $0\leq k\leq |Y|$,

\begin{algorithm}[ht]
\caption{SNH-SB Optimization}
\label{alg:core}
\begin{algorithmic}[1]
\Require Given $X,Y,E,E_R,\rho$, aggregated cost bound $\beta$, aggregation function $g\in \{\Ut(\cdot), \Si(\cdot)\}$ and weight function $w:E\cup E_R\rightarrow R^+$

\Ensure A relaxable set $\hat{E}$
\State $G \leftarrow (X,Y,E\cup E_R)$
\State $k_{low}\leftarrow 0, k_{high}\leftarrow |Y|$
\While{$k_{low}\leq k_{high}$}
\State $k = \left\lfloor \frac{k_{\text{low}} + k_{\text{high}}}{2} \right\rfloor$
\State Construct $G_k$ by adding $k$ dummy agents to $G$ and connecting them to all resources of $Y$ with edges $E_k$
\ForAll {$e\in E_k$}
\State $w(e)=1+\sum_{e\in {E \cup E_R}}w(e)$.
\EndFor
\State Compute a maximum weighted matching $M_k$ of $G_k$
\If {$g(M_k\setminus (E \cup E_k))\leq \beta$}
\State $E_{min} \leftarrow E_k$, $M_{min} \leftarrow M_k$
\State $k_{high}\leftarrow k-1$ 
\Else 
\State $k_{low}\leftarrow k+1$
\EndIf 
\EndWhile
\State $\hat{E}=M_{min}\setminus (E\cup E_{min})$
\State \Return $\hat{E}$
\end{algorithmic}
\end{algorithm}

It is important to note that the weights $w$ defined in 
Algorithm~\ref{alg:core} are only used by the facilitator to determine the relaxation $\hat{E}$.  The allocation itself is performed by a third party \emph{with no weights}.  Thus, the facilitator needs to devise appropriate weights that will guarantee the desired properties in the subsequent \emph{weightless} 
allocation. Devising these weight functions and proving their validity for various participation guarantees is the topic of subsequent sections.

\paragraph{The functions $u(\cdot)$.}
We define the function $u(\cdot)$ differently for each aggregation 
function. For the size function, $\Si(\cdot)$, $u_1(e)\equiv 1$ has 
value~1 for each relaxation~$e$. In this way, each relaxed edge has a
smaller weight compared to the weight of an original edge in the maximum
weighted matching computed in line~8 of Algorithm~\ref{alg:core}. 
For the total cost function~$\Ut(\cdot)$, we define~$u_2(e)$ to be the
discomfort level of relaxing~$e$.
Summing up the contributions of all the~$u_2(e)$ values to the weight
function, results in a factor that is equal to the total cost, divided by
some constant.

We define the weight function $w$ as:
\begin{align*}
        w(e) = \begin{cases}
            (|X|+1)^2 & e\in E \\
            |X|+1-\frac{u(e)}{\max_{e'\in E_R}\{ u(e')\}} & e\in E_R
        \end{cases}
\end{align*}

The main motivation behind this definition is to assign high weights to the original edges, ensuring that the maximum matching of the graph includes them as a subgraph even after any relaxations. Building on this, the function aims to increase the overall size of the matching by incorporating relaxable edges according to their costs.

We have the following result
for Algorithm~\ref{alg:core}.
%solves the SNH-SB optimization problem.
%%
\begin{theorem}
Algorithm~\ref{alg:core} with the weights function $w(\cdot)$ solves the 
SNH-SB optimization problem with aggregate function 
$g\in \{\Ut(\cdot),\Si(\cdot)\}$ and bound $\beta$.
\label{pro:pcfair}
\end{theorem}

The runtime of one iteration of the loop, starting in line 3 of Algorithm~\ref{alg:core}, is dominated by the weighted matching calculation, which can be done in $O(\max(|X|,|Y|)^3)$ time using the Hungarian algorithm. The loop repeats $O(\log(|Y|)$ times and therefore the total running time is $O(\log(|Y|) * \max(|X|,|Y|)^3)$. 

We now outline the proof of the theorem. 
\iftoggle{arxiv}
{The full proof appears in Section~\ref{sec:appalg} of the appendix.}
{The full proof appears in an extended version
of this paper~\cite{Trabelsi-etal-ecai-2025}.} 
Recall that for an edge set~$F$, 
$\mu(F)$ is the size of a maximum allocation of~$F$. Recall also the definition of a minimal relaxation~(Section~\ref{sec:prelims}).
To provide the strong no harm strong benefit guarantee (SNH-SB), we must
provide two participation guarantees, namely no-harm and benefit to relaxers, 
even when some agents do not follow the facilitator's advice. 
From definitions of the SNH-SB guarantees, it follows that for each 
SNH-SB relaxation~$\hat{E}$, any subset of $\hat{E}$ also provides the 
SNH-SB guarantees. Therefore, any minimal relaxation 
$\hat{E}'\subseteq \hat{E}$ for which $\mu(\hat{E}')=\mu(\hat{E})$ should also be
SNH-SB. Therefore, it suffices for Algorithm~\ref{alg:core} to consider relaxations which are  minimal. We will show in Lemma~\ref{lem:inc} that for all minimal SNH-SB relaxations, $\mu(E\cup \hat{E}) = \mu(E) +|\hat{E}|$. Therefore, any maximum matching of $(X,Y,E\cup \hat{E})$ must contain a maximum matching of $(X,Y,E)$ as a subgraph.  
We filter out all solutions without these properties by modifying the weight function and assigning significantly greater weights to the $E$ edges.

We show in Lemma~\ref{lem:mkpc} that Algorithm~\ref{alg:core} indeed filters out all non SNH-SB relaxations (see the proofs of Lemma~\ref{lem:mkpc}(1) and (2)) and that the returned relaxation maximizes the allocation size 
\iftoggle{arxiv}
{(see Lemmas~\ref{lem:mkpc}(3) and (4) and also Lemma~\ref{lem:gequ} in Section~\ref{sec:appalg} of the appendix.}
{(see Lemmas~\ref{lem:mkpc}(3) and (4) and also Lemma~\ref{lem:gequ} in the extended version of the paper~\cite{Trabelsi-etal-ecai-2025}).}
In the following lemmas, we denote $k_{\min}$ as the value of $k$, used in line 14 to define $M_{\min}$ and $E_{min}$. 
The values of $M_{\min}$ and $E_{min}$ in line 14 are denoted as $M_{k_{\min}}$ and $E_{k_{min}}$.

\begin{lemma}
 The set, $\hat{E}$ returned by  Algorithm~\ref{alg:core} has the following properties:
\begin{enumerate}[label={(\arabic*)}]
    \item It provides the strong benefit for relaxers.
    \item It provides the strong no harm guarantee.
    \item For any SNH-SB set  $\hat{F}\subseteq E_R$,  $|M_{k_{min}}|\geq \mu(E\cup E_{k_{min}}\cup \hat{F})$.
    \item For any SNH-SB set $\hat{F}\subseteq E_R$ and for any aggregation function $g$, if $|M_{k_{min}}|=\mu(E\cup E_{k_{min}} \cup \hat{F} )$, then $\hat{E} \preceq \hat{F}$.        
\end{enumerate}
\label{lem:mkpc}
\end{lemma}

Set $w^*=|X|+1$ and let $\hat{E}=M_{k_{min}}\setminus E$ 
be the relaxation returned by the algorithm.

\begin{proofideaof}[lemmas ~\ref{lem:mkpc}(1) and ~\ref{lem:mkpc}(2)]
Consider $\hat{F}\subseteq \hat{E}$. Let $\bar{M}=M_{k_{min}}\cap E$.  
We will first show that  $|\bar{M}_{\hat{F}}| =k_{min} + \mu(E)+|\hat{F}|$.  So, (i) $|\bar{M}_{\hat{F}}\cap E|=\mu(E)$, and (ii) $|\bar{M}_{\hat{F}}\cap \hat{F}|=|\hat{F}|$.  So, by (i) $\bar{M}_{\hat{F}}\cap E$ is a maximum matching of $(X,Y,E)$, so $\Gamma(E)\subseteq \bar{M}_{\hat{F}}$, and by (ii) $\hat{F}\subseteq \bar{M}_{\hat{F}}$. 
Since it is correct for any $\hat{F}$, $\hat{E}$ has the no-harm and benefit for relaxers SNH-SB guarantees with respect to $(X\cup X_{k_{min}},Y,E\cup E_{k_{min}}\cup \hat{E})$.

By construction of $E_{k_{min}}$, all agents $X_{k_{min}}$ participate in all maximum matchings of $(X\cup X_{k_{min}},Y,E\cup E_{k_{min}}\cup \hat{F})$. 
It can be shown that $\hat{E}$ has also the SNH-SB guarantees also with respect to $(X,Y,E\cup\hat{E})$ as requested.
\end{proofideaof}

\begin{proofideaof}[lemma ~\ref{lem:mkpc}(3)]

For purposes of contradiction, suppose that there is another SNH-SB set $\hat{F}$, for which the maximum matching of $(X\cup X_{k_{min}},Y,E \cup \hat{F} \cup E_{k_{min}})$ is of greater size. 

We set $\hat{F}'\subseteq \hat{F}$ as an inclusion minimal subset of $\hat{F}$ for which $\mu(\hat{F}')=\mu(\hat{F})$.
In particular, since $\hat{F}'$ is inclusion minimal,  $\mu(E\cup E_{k_{min}}\cup{\hat{F}})>\mu(E\cup E_{k_{min}}\cup{\hat{F}}\setminus \{e\}))$.
So, the conditions of Lemma \ref{lem:inc} hold for $(E\cup E_{k_{min}})$ and $\hat{F}'$. 
Let $\bar{M}_{\hat{F}'}$ be a maximum matching of $(X,Y,E\cup E_{k_{min}}\cup{\hat{F}})$. 
So, since $M_{k_{min}}$ is a matching of $(X\cup X_{k_{min}},Y,E\cup E_{k_{min}} \cup \hat{E})$, $|\bar{M}_{\hat{F}'}|> |M_{k_{min}}|$.
So,
\begin{align*}
    |\bar{M}_{\hat{F}'}\cap \hat{F}'| ~=~ & |\bar{M}_{\hat{F}'} \setminus E | = |\bar{M}_{\hat{F}'}| - |\bar{M}_{\hat{F}'}\cap E| \\
    ~=~ &  |\bar{M}_{\hat{F}'}| - \mu(E) \;\; \text{by Lemma~\ref{lem:inc}} \\
    ~>~ & |M_{k_{min}}| - \mu(E) \;\; \text{by assumption} \\
    ~=~ & |M_{k_{min}}| - |M_{k_{min}}\cap E| \;\; \text{by Lemma~\ref{lem:mwc}} \\
    ~=~ & |M_{k_{min}} \cap \hat{E}| +|M_{k_{min}} \cap E_{k_{min}}|
\end{align*}
\iftoggle{arxiv}
{We will show in Section~\ref{sec:appalg} of the appendix that}
{Using this equation, it can be shown that}
\begin{align}
w(M_{k_{min}}) ~<~ w(\bar{M}_{\hat{F}'}) 
\end{align}
This contradicts the maximality of $w(M_{k_{min}})$.
\end{proofideaof}

\begin{proofof}[lemma ~\ref{lem:mkpc}(4)]
Let $\hat{F}$ be a SNH-SB relaxation with maximum allocation when adding $k_{min}$ dummy agents and for which $g(\hat{F})\leq \beta$.

First note that since $\hat{F}$ is SNH-SB so is any subset of $\hat{F}$. So, for any $e\in \hat{F}$, $\mu(E\cup{\hat{F}})>\mu(E\cup{\hat{F}\setminus \{e\}})$, or else  $\hat{F}\setminus \{e\}$ is a SNH-SB set with allocation of the same size that preceding $\hat{F}$ in $\preceq$ (Since $u(\cdot)$ is strictly positive).

For $e\in E_R$, set $u'(e)=\frac{u(e)}{\max_{e'\in E_R}\{ w(e')\}}$.  So, $u'$ is also a representation of $\preceq$, and
\begin{align*}
        w(e) = \begin{cases}
            (w^*)^2  & e\in E \\
            w^*-u'(e) & e\in E_R
        \end{cases}
\end{align*}
We have already established that
\begin{align*}
\mu(E\cup E_{k_{min}} \cup \hat{E}) &=\mu(E)+ k_{min}+ |\hat{E}|.   
\end{align*}
And by Lemma~\ref{lem:inc}
\begin{align*}
\mu(E\cup E_{k_{min}}\cup \hat{F})=\mu(E)+ k_{min}+ |\hat{F}|.  
\end{align*}
So, since both $\hat{E}$ and $\hat{F}$ have allocation of the same size, $|\hat{E}|=|\hat{F}|$.
By Lemma~\ref{lem:mwc}
(for the first equality) and Lemma~\ref{lem:inc} (for the second) we have 
\iftoggle{arxiv}{(Details are in Section~\ref{sec:appalg} of the appendix.)}
{(see the extended version of the paper~\cite{Trabelsi-etal-ecai-2025} for details)}:
\begin{align*}
w(M_{k_{min}})& ~=~\\ &w(M_{k_{min}}\cap    
    E_{k_{min}}) -u'(\hat{F})+|\hat{F}|w^*+\mu(E)\cdot (w^*)^2 
\end{align*}
By construction $w(M_{k_{min}})\geq w(\bar{M}_{\hat{F}})$.  So, since $|\hat{E}|=|\hat{F}|$, $u'(\hat{E})\leq u'(\hat{F})$.  So, $\hat{E}\preceq\hat{F}$.
\end{proofof}

The two lemmas below are used 
 for proving Lemma~\ref{lem:mkpc}:
\begin{lemma}
\label{lem:inc}
Let $\hat{E}$ be a minimal SNH-SB relaxation. 
Then $\mu(E\cup \hat{E}) = \mu(E) +|\hat{E}|$. 
\end{lemma}

The following three results are used in proving
Lemma~\ref{lem:inc}. In
these results, $G=(X,Y,E)$ is a bipartite graph.

\begin{description} 
\item{(1)}
Let $x\in \Gamma(E)$. Let $(x,y)$ be an edge in $E_R$.
Then $\mu(E)=\mu(E\cup \{(x,y)\})$.

\begin{proofideaof}[(1)]

We first note that since $E\subseteq E\cup \{(x,y)\}$, it follows that $\mu(E)\leq \mu(E\cup \{(x,y)\}$. 
We assume for contradiction that $\mu(E)<\mu(E\cup \{(x,y)\})$ and show that it contradicts the participation of $x$ in all maximum matchings of $G=(X,Y,E)$.
\end{proofideaof}

\smallskip 

\item{(2)}
Let $x$ be an agent that 
does not participate in all maximum matchings of $G$ and
$(x,y)\in E_R \notin E$.
If $x$ participates in all maximum matchings of 
$(X, Y, E\cup \{(x,y)\})$,  then 
$\mu(E\cup \{(x,y)\})=\mu(E)+1$.

\begin{proofof}[(2)]
    We assume that $x$ participates in all maximum matchings of $(\agset, \rset, E\cup \{(x,y)\})$ and prove that $\mu(E)+1=\mu(E\cup \{(x,y)\})$.
    Since $(X,Y,E\cup \{(x,y)\})$ is supergraph of $G$, 
    $\mu(E)\leq \mu(E\cup \{(x,y)\})$. On the other hand, from definition of matching,
    $\mu(E\cup \{(x,y)\})\leq \mu(E)+1$.
    Therefore, either $\mu(E\cup \{(x,y)\}) = \mu(E)$ or $\mu(E\cup \{(x,y)\}) = \mu(E)+1$. Assume for contradiction that $\mu(E\cup \{(x,y)\}) = \mu(E)$. In this case, all maximum matchings of $G$ are also maximum matchings of $(X,Y,E\cup \{(x,y)\})$. From the lemma's definitions, there is a maximum matching of $G$ in which $x$ does not participate. Since it is also a maximum matching of $(X,Y,E\cup \{(x,y)\})$, $x$ does not participate in all maximum matchings of $(X,Y,E\cup \{(x,y)\})$, a contradiction.
\end{proofof}

\smallskip 

\item{(3)} 
A node $x \in X$ does not participate in all maximum matchings of $G$ iff there is an even-length alternating path between $x$ and a free node with respect to any maximum matching of $G$.
(This well-known result in matching
theory~\citep{Lovasz-Plummer-1986}; 
\iftoggle{arxiv}{it is proved as
Claim~\ref{cl:result_3} in the appendix for completeness. The definitions for alternating paths and free nodes are also given in the appendix.}
{(A proof of this result and the definitions for alternating paths and free nodes are also given in the extended version 
of the paper~\cite{Trabelsi-etal-ecai-2025}.)}
\end{description}
%\textcolor{red}{(The terms ``augmenting path'' and ``free node''
%have not been defined at this point. It may be useful to
%add these definitions in an early section of the Appendix and
%refer the reader to that section. -- Ravi)} \ycomment{Fixed.}

\hspace*{0.25em}
We use the above three results to prove 
Lemma~\ref{lem:inc}.

\begin{proofof}[Lemma~\ref{lem:inc}]
We recall that an edge set $\hat{E}$ is minimal if for any edge $e\in \hat{E}$,
$\mu(E\cup{\hat{E}\setminus \{e\}})<\mu(E\cup{\hat{E}})$.
The proof is by induction on $|\hat{E}|$.
The basis is for $|\hat{E}|=1$, where the lemma is fulfilled trivially.
We now assume that the lemma holds for $|\hat{E}|=q$ 
and prove it for $|\hat{E}|=q+1$.

Let $|\hat{E}|$ be a SNH-SB relaxation and assume that for any edge $e\in \hat{E}$,
$\mu(E\cup{\hat{E}\setminus \{e\}})<\mu(E\cup{\hat{E}})$.
Let $x$ be any agent in $X(\hat{E})$, and let $(x,y)\in\hat{E}$ be the corresponding edge in $\hat{E}$. 
It follows that:
\begin{align}\label{eqn:lem_inc1}
\mu(E\cup{\hat{E}\setminus \{(x,y)\}}) + 1 = \mu(E\cup{\hat{E}}) 
\end{align}

\smallskip

\noindent
\underline{Case 1:}~ The set $\hat{E}\setminus \{(x,y)\}$ is minimal: Here,
by induction hypothesis, $\mu(E\cup{\hat{E}\setminus \{(x,y)\}}) = \mu(E)+q$.
Therefore, if we merge this condition with Equation~\eqref{eqn:lem_inc1}, we get 
$\mu(E\cup{\hat{E}}) = \mu(E\cup{\hat{E}\setminus \{(x,y)\}}) +1 = \mu(E)+q+1$ as required.

\smallskip
\noindent
\underline{Case 2:}~ The set $\hat{E}\setminus \{(x,y)\}$ is not minimal:
In this case, there is an edge $(x',y')\in\hat{E}\setminus \{(x,y)\}$ such that
$\mu(E\cup \hat{E}\setminus \{(x,y)\})=\mu(E\cup \hat{E}\setminus \{(x,y),(x',y')\})$.
Let $M$ be a maximum matching of $(X,Y,E\cup \hat{E}\setminus \{(x,y)\})$ without $(x',y')$. (Such a matching exists because of 
Equation~\eqref{eqn:lem_inc1} above.) 

We note that since $\mu(E\cup{\hat{E}})>\mu(E\cup \hat{E}\setminus \{(x,y)\})$, Result~(2) above
implies that $x$ must not participate in all maximum matchings of $(X,Y,E\cup \hat{E}\setminus \{(x,y)\})$.
Therefore, it follows from Result~(3) above that there is 
an even-length alternating path between $x$ and a free agent concerning any maximum matching of $(X,Y,E\cup \hat{E}\setminus \{(x,y)\})$.
On the other hand, since $\hat{E}$ has the SNH-SB benefit to relaxers, $x'$ must participate in all maximum matchings of $(X,Y,E\cup \hat{E}\setminus \{(x,y)\})$.
Therefore, by Result~(3) above, $x'$ does not appear in any even length alternating path between $x$ and a free vertex in any maximum matching of $(X,Y,E\cup \hat{E}\setminus \{(x,y)\})$.

Since $M$ is also a maximum matching of $(X,Y,E\cup \hat{E}\setminus \{(x,y),(x',y')\})$, an alternating path of even length between $x$ and a free vertex of $M$ exists in $(X,Y,E\cup \hat{E}\setminus \{(x,y),(x',y')\})$ and therefore $x$ does not participate in all maximum matchings of $(X,Y,E\cup \hat{E}\setminus \{(x,y),(x',y')\})$.

We now add $(x,y)$ to $(X,Y,E\cup \hat{E}\setminus \{(x,y),(x',y')\})$ and get $(X,Y,E\cup \hat{E}\setminus \{(x',y')\})$.
If $x$ participates in all maximum matchings of $(X,Y,E\cup \hat{E}\setminus \{(x',y')\})$, then by Result~(1) above, the matching size must increase, contradicting the fact that the original set is minimal. Otherwise, $x$ does not participate in all maximum matchings after $x$ relaxes its restrictions, contradicting the SNH-SB.
Therefore, we conclude that no such $x'$ exists, and we are done. 
\end{proofof}

Lemma~\ref{lem:mwc} below is based on the definition of $w(\cdot)$.
\begin{lemma}
$|M_{k_{min}}\cap E|=\mu(E)$.
\label{lem:mwc}
\end{lemma}

\subsection{ Strong no harm weak benefit guarantee (SNH-WB)}
\label{sec:mc}
As in the SNH-SB problem, the SNH-WB problem can also be solved using Algorithm~\ref{alg:core}.  
Intuitively, as in the right column of Figure~\ref{fig:examplefc}, it is possible that a maximum matching of $(X, Y, E \cup \hat{E})$ does not have a maximum matching of $(X, Y, E)$ as a subgraph. Therefore, the weights of the $E$ edges must be reduced. However, these weights must remain greater than those of $\hat{E}$ to ensure that the WB guarantee is preserved.
%The hierarchy between WNH-WB and SNH-WB implies that not all feasible solutions for WNH-WB are also feasible for SNH-WB and hence we adjust the weights for filtering them out.

Note that if the set $\hat{E}$ returned by the algorithm contains an edge $\{x,y\}\in \hat{E}$ such that $x\in \Gamma(E)$, then $x$ cannot participate in all maximum matchings of $(X,Y,E\cup \hat{E}\setminus \{x,y\})$ and therefore, the SNH-WB no-harm guarantee does not hold. That is so since otherwise, $\mu(E\cup \hat{E}\setminus \{x,y\})=\mu(E\cup \hat{E})$ and the maximum weighted matching in line 8 won't contain the relaxed edge $\{x,y\}$, contradicting the fact that $\{x,y\}\in \hat{E}$. Therefore, we assign a negative weight to all edges with agent in $\Gamma(E)$ and this way, we keep the no-harm guarantee even when some agents of $\hat{E}$ do not relax.  
The updated weight function is defined as follows:
\begin{align*}
        w(e) = \begin{cases}
            -1 & e=(x,y)\in E \cup E_R, x\in \Gamma(E) \\
            (|X|+1) & e=(x,y)\in E, x\notin \Gamma(E)\\
            |X|+1-\\ \;\;\;\; \frac{u(e)}{\max_{e'\in E_R}\{ u(e')\}} & e=(x,y)\in E_R, x\notin \Gamma(E)
        \end{cases}
\end{align*}
In this function, the weights of the original edges are smaller as it is not required to have a maximum matching from original edges as a subgraph of the resulting matching; however, it is required that the guaranteed agents won’t be asked to relax their restrictions. This way, if they do not relax, no harm will be caused.

Note that the set $\Gamma(E)$ can be computed in polynomial time using the method of~\cite{irving2006rank}.
The functions $u$ that are used are identical to those used in the SNH-SB case. We conclude with the following theorem:
\begin{theorem}
Algorithm~\ref{alg:core} with the weights function $w(\cdot)$ solves the
SNH-WB optimization problem with an aggregate cost function $g$ (see Theorem~\ref{pro:pcfair}) and a bound $\beta$.
    \label{pro:mc}
\end{theorem}

%% file: many.tex
\section{Many-to-One and One-to-Many Allocations}
\label{sec:mw}
%\textcolor{red}{(I added the word ``Allocations''
%to the title of this section. Please check. -- Ravi)} %\ycomment{fine with me}

Up to now, we have considered allocating a single resource to each agent. 
%However, in practice, it is often not the case. 
In many scenarios like allocation of classrooms to courses~\cite{trabelsi2023resource}, a single agent might need multiple resources, e.g., for providing a course with enough classrooms for
holding an exam where students need to be far from each
other.  In other scenarios, a single resource might be shared among multiple agents, e.g., for providing multiple courses with very few participants with an appropriate classroom for holding an exam together.
This section formulates an optimization problem for many-to-one allocations and presents a polynomial time algorithm to solve it.
The one-to-many case where a resource can be
shared by multiple agents can be done similarly.

The problem is similar to Definition~\ref{def:prob} but with an addition: a function $d:X\rightarrow \mathbb{N}$ that assigns to each agent its desired amount of resources is added to the problem inputs.
Note that $\mu(\cdot)$ is the size of a maximum valid allocation with respect to $d$ (not necessarily a maximum matching).

Note when moving to the many-to-one setting, the aggregation functions are naturally applied to the set of relaxed edges (not relaxing agents). In particular, the size aggregation cost function
bounds the overall number of relaxed edges.

We propose Algorithm~\ref{alg:manyone} that uses Algorithm~\ref{alg:core}. Algorithm~\ref{alg:manyone} first duplicates the agents that need more than one resource and adds some relevant edges. It then runs  Algorithm~\ref{alg:core} on the modified graph and returns its results. 
Lemma~\ref{lem:many} is crucial for showing that this modification does not affect the desired participation guarantees. 
\begin{lemma}
    Let $G=(X,Y,E)$ be a graph and let $x,x'\in X$ be agents. Set $Y_x$ and $Y_{x'}$ as the resources adjacent to $x$ and $x'$ in $E$.
    If $Y_x \subseteq Y_x'$ and $x\in \Gamma(G)$,  then $x'$ is also in $\Gamma(G)$.
    \label{lem:many}
\end{lemma}

\begin{proofof}[lemma~\ref{lem:many}]
    Assume for the purposes of contradiction that
    $x' \not\in \Gamma(G)$. Let $M$ be a maximum matching of $G$ without $x'$.
    Since $x\in \Gamma(G)$, there is $y$ for which $\{x,y\}\in M$.
    Since $Y_x \subseteq Y_x'$, the edge $\{x',y\}\in E$. therefore, we can substitute the edges and have a maximum matching of $G$ without $x$, and we have a contradiction.
\end{proofof}

\begin{algorithm}

\begin{algorithmic}[1]
\caption{\label{alg:manyone}}
\Require Given, $X,Y,E,\rho$, aggregated cost bound $\beta$, aggregation function $g\in \{\Ut(\cdot),\Si(\cdot)\}$, demand function $d:X\rightarrow N$ and a weight function $w:E\cup E_R\rightarrow R^+$

\Ensure A relaxable set $\hat{E}$
\State $X',Y',E'\leftarrow \emptyset$

\ForAll{agents $x\in X$}
\State $X' = X'\cup \{x_{j}|0 \leq j < d(x)\}$.
\EndFor
\ForAll {resources $y\in Y$} 
\State $Y' = Y'\cup \{y\}$.
\EndFor
\State $E'=\{\{x_{j},y\}|\{x,y\}\in E\}$,  
\State $E_R'=\{\{x_{j},y\}|\{x,y\}\in E_R\}$,  
\State $\forall{\{x_{j},y\}}\in E_R'$, $\rho(\{\{x_{j},y\}) = \rho(\{x,y\})$,  

\State $\forall{\{x_{j},y\}}\in E'\cup E_R',  w'(x_{j},y)=w(x,y)$

\State Run Algorithm~\ref{alg:core} with $X',Y',E',E_R',\rho', \beta, g$ and $w'$.
Let $\hat{E}'$ be the returned edge set.
\State Construct a set $\hat{E}$ from $\hat{E}'$ by substituting all instances $x_j$ of $x$ by $x$ itself.
\State \Return $\hat{E}$
\end{algorithmic}
\end{algorithm}
We conclude with the following:

\begin{theorem}
Algorithm~\ref{alg:manyone} with a weight function $w(\cdot)$ corresponds to a participation guarantee PG $\in$ \{SNH-SB, WNH-WB, SNH-WB\} (e.g., the weight functions in Section~\ref{sec:algorithms}) solves the many-to-one PG optimization problem with an aggregate cost function $g$ (see Theorem~\ref{pro:pcfair}), a bound $\beta$ and a demand function $d(\cdot)$.
\label{pro:pmo}
\end{theorem}
\begin{proofof}[Theorem~\ref{pro:pmo}]

Note that Algorithm~\ref{alg:manyone} uses Algorithm~\ref{alg:core} as a subprocedure. The correctness of Algorithm~\ref{alg:core} was proved in the previous sections.
We conclude that Algorithm~\ref{alg:core} returns a relaxation set with a maximum sized allocation that has the relevant participation guarantee  with the relevant aggregate function and bound.
It remains to show that relevant participation guarantee is preserved when merging the duplicated agents as we do in line 11 of Algorithm~\ref{alg:manyone}. 

Lemma~\ref{lem:many} implies that if one instance of an agent is guaranteed to be matched so is for all other instances. Therefore, the different guarantees for relaxers and the different no-harm guarantees provided by Algorithm~\ref{alg:core} to some instances of an agent hold for all and therefore for the agent itself. This concludes the proof of Theorem~\ref{pro:pmo}.
\end{proofof}

The runtime of Algorithm~\ref{alg:manyone} is determined by the runtime of Algorithm~\ref{alg:core} on line 10. The key difference lies in the size of \( X' \), which is \( \sum_{x \in X} d(x) \). The overall complexity is therefore \( O(\log(|Y|) \cdot \max(\sum_{x \in X} d(x), |Y|)^3) \).

%% file: experiments.tex
\section{Experiments}
\label{sec:experiments}
In this section, we discuss the results of employing a facilitator, as described in this paper, on three different real-world datasets. 
Experiments were conducted in all three allocations settings (one-to-one, one-to-many and many-to-one).  We compared the WNH-WB, SNH-WB and the SNH-SB participation guarantees and the extent to which these different participation guarantees affect the resultant allocation size.\footnote{The code and data for running the experiments are available at~\cite{trabelsi2025facilitating}}.

\paragraph{Datasets.}
We ran our experiments on three real-world datasets. The \textbf{courses
and classes dataset (COURSE)}~\cite{trabelsi2022resource} contains 154 courses and
144 classrooms, with the goal of matching courses to appropriate
classrooms.  The attributes in this dataset are room capacity, location, availability of accommodations for physical
disability, and availability of accessories for hearing disability. 
The \textbf{students lab dataset (LAB) }~\cite{trabelsi2023resource} consists of a lab
with 31 students that should be matched to seats in 14 lab rooms. 
Room attributes include proximity to some locations in the lab (e.g., advisor's
room, restrooms, kitchen), capacity, and strength of the Wi-Fi signal. 
The \textbf{children activities
dataset (CHILD)} \cite{varone2019dataset} consists of 653 children and 533 available activities over a period of
several weeks. We focused on one week during the autumn vacation in the
Swiss municipality of Morges. During that week, 249 activities were offered.
Attributes included minimum and maximum allowed age and children's
priorities over the activities. Note that the CHILD dataset has significantly more agents and relatively less available resources, so the overall match rate and the potential for  improvements by the facilitator are much more limited.

Appropriate \textbf{discomfort} functions were defined for the datasets. 
In the LAB dataset, the attributes were explicitly rated on a scale of 1--5, and the discomfort function was defined accordingly, with rating 1 taken as no discomfort, and 2-5 having discomfort of 1-4, respectively. 
The discomfort for relaxing an edge is the sum of the discomforts for all its attributes.
For the COURSE
dataset, the discomfort was estimated from the nature of the constraints. These estimates were informed by consultations with a university administrator and faculty members.
The availability of accommodations for physical disability and the availability of accessories for hearing disability have a crucial impact on the discomfort level. The distance from the desired location determines the discomfort incurred by the room's location. Finally, the discomfort of the room capacity is more significant as the difference from the desired capacity is more significant. 

The CHILD dataset has attributes of both types (explicit rating and not), so both methods were employed for determining the discomfort function. Children rated some activities, and their ratings  contributed to the discomfort in a way similar to the attributes in the LAB dataset. If a child is younger than the minimum age for the activity, the difference between her age and the minimum age of the activity was considered for determining the discomfort. The contributions to the discomfort when the child is older than the maximum allowed age for the activity were computed similarly.
%\textcolor{red}{(Will it be useful to add more details about
%the cost functions in the supplement? If so, we can add a 
%note here to refer the reader to the supplement. -- Ravi)}
%\ycomment{Added a reference.}

%If the costs are too high, it is unlikely that an agent will follow the facilitator's advice. We investigated different cost thresholds  - were costs above the threshold considered infinite - and analyzed the performance of the algorithm for these thresholds.

\paragraph{Maximum possible match sizes when all agents comply.}
\label{sec:msize}
Here, we compare the WNH-WB, SNH-WB and SNH-SB relaxations for the special case where all agents comply. Since a SNH-SB relaxation is optimized to handle the case when not all agents comply, this stricter condition might impact the matching size when compared with WNH-WB, even when all agents comply. The SNH-WB is somewhere between the SNH-SB and the WNH-WB in the courses dataset and is comparable with SNH-SB in the two other datasets. 
%To measure the cost of the SNH-SB guarantee and, we should compare the WNH-WB and SNH-SB relaxations in two scenarios: when all agents comply and when they do not. The first comparison is important when we expect that 
%% That is since the SNH-SB guarantee imposes more constraints on the solution space than WNH-WB.  We wanted to check to what extent these extra constraints decrease the resulting match size. 
%This question arises both when all agents comply and when they do not.      
Figure~\ref{fig:msize} depicts the resultant matching sizes for WNH-WB, SNH-WB and SNH-SB relaxations when all agents comply. The baseline bar is for when no relaxation is allowed. First, note that relaxation does indeed increase the match size considerably, even for low bounds.  In the CHILD dataset, a 5\% increase in the relaxation cost bound allows 32 more children to be matched. 
%\textcolor{red}{(It may be better to state this as follows:
%A 5\% increase in the relaxation cost budget allows 32 more %children to be matched. -- Ravi)} \ycomment{fixed. }
Comparing the results of WNH-WB, SNH-WB and SNH-SB, we see that while SNH-SB and SNH-WB do provide  better results sometimes in smaller matchings, the difference is relatively small.%\ycomment{Sentence was not clear, please check now.} 

\begin{figure}[htb]
\centering

\includegraphics[width=0.47\columnwidth]{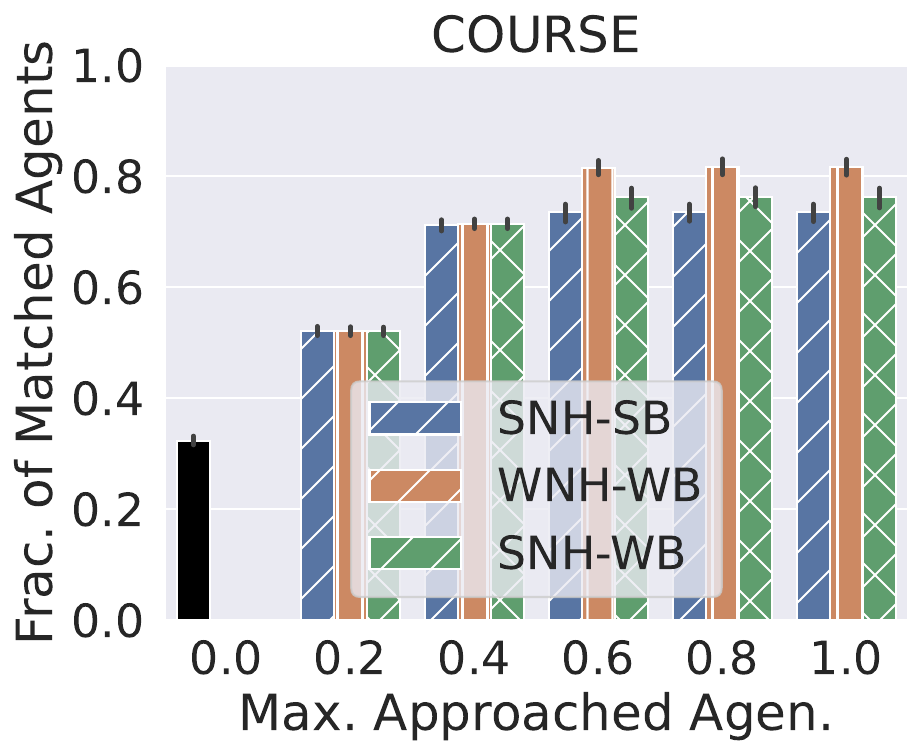}
\includegraphics[width=0.45\columnwidth]{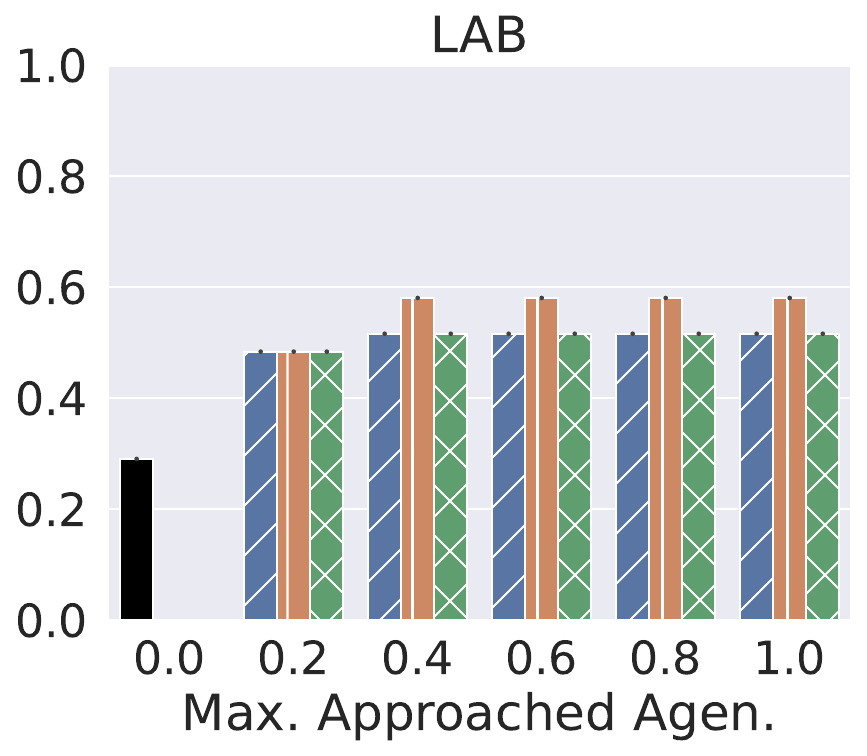}
\includegraphics[width=0.45\columnwidth]{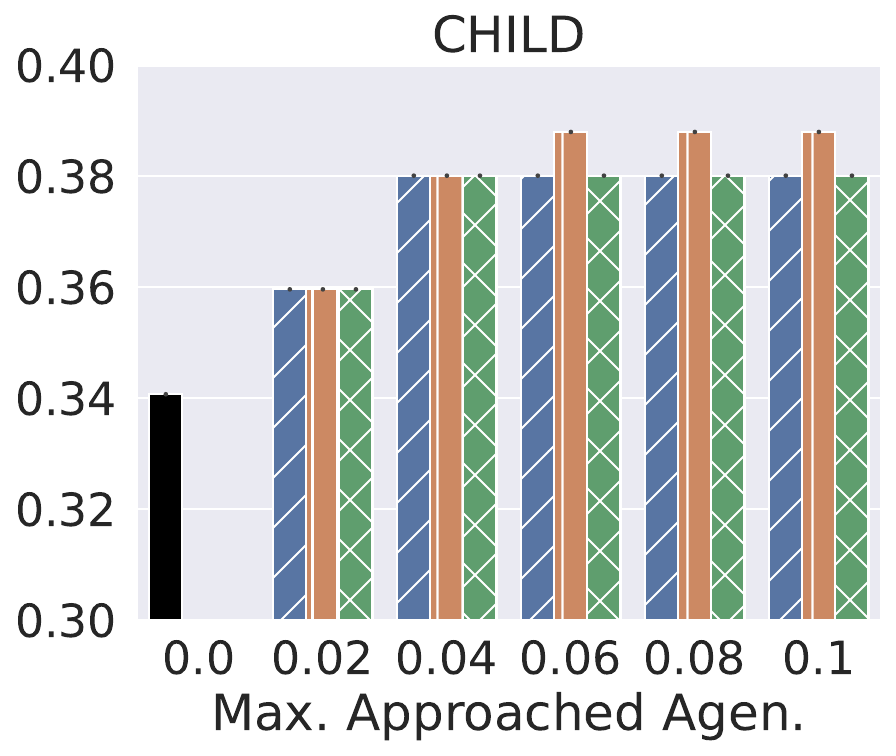}

\caption{Fraction of matched agents as a function of the bound on allowed relaxations, for the SNH-SB, SNH-WB and WNH-WB guarantees. The scale and limits of the $y$ axis are different for the different datasets. The baseline (in black) corresponds to no relaxations or a cost bound of zero.\label{fig:msize}}

\end{figure}

\paragraph{Match sizes when not all agents comply.}
Figure~\ref{fig:relaxing_agents} depicts the matching size as a function of the number of complying agents, with no bound restrictions. Here, we first computed the full SNH-SB, SNH-WB and WNH-WB relaxations, using the methods devised in the earlier sections, randomly selected a set of complying agents (of different sizes), and then computed the maximum matching with the resultant relaxation. Each point on the graph is an average of ten such random samples.   
%\ycomment{I removed the notion of threshold here. I think that it is not very clear to the reader why is it needed in addition to the "bound"} \textcolor{red}{(I think removing the notion of threshold here is a good idea. -- Ravi)}
We can see that SNH-SB performs better with low and moderate compliance rates, while WNH-WB performs better when more agents comply.  
%\textcolor{red}{(Here, ``very high compliance''
%may confuse a reader since WNH-WB is for full compliance. -- %Ravi)}\ycomment{please check}
The performance of SNH-WB is comparable to that of SNH-SB in most cases and it is  between SNH-SB and WNH-WB for the COURSE dataset.

Hence, we propose using SNH-SB when it is expected that several agents will not comply. However, if all agents are known to comply, WNH-WB is preferable. In the intermediate case, where we expect very high compliance, the facilitator should weigh the costs of losing the SNH-SB or the SNH-WB guarantees and the dissatisfaction of some agents against the benefits of increasing the allocation size.

\begin{figure}[htb]
\centering

\includegraphics[width=0.46\columnwidth]{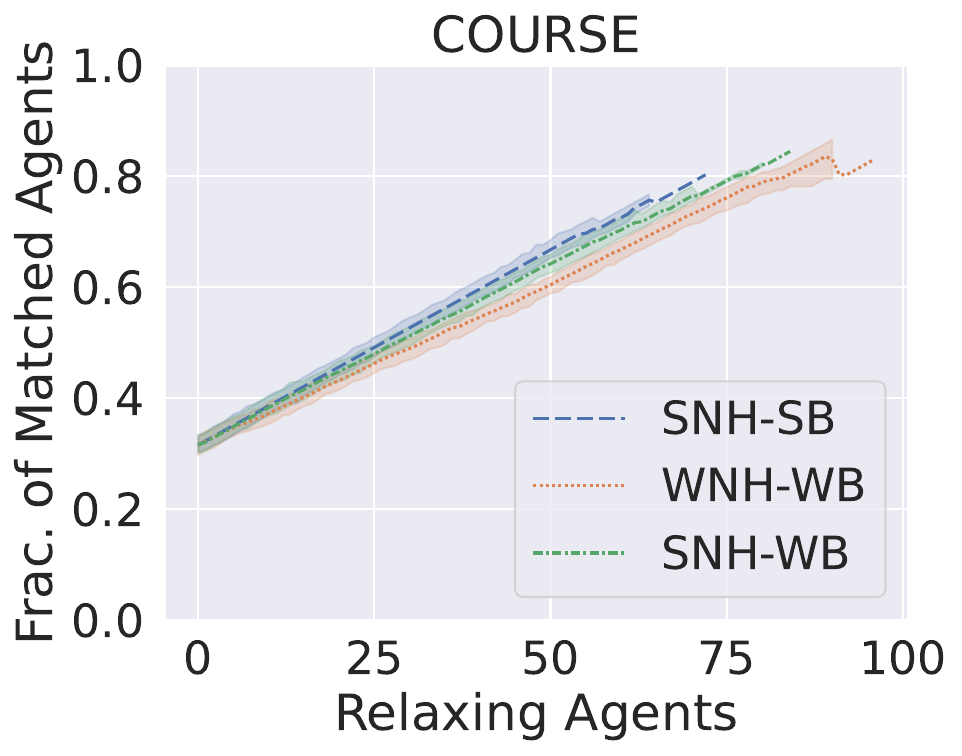}
\includegraphics[width=0.42\columnwidth]{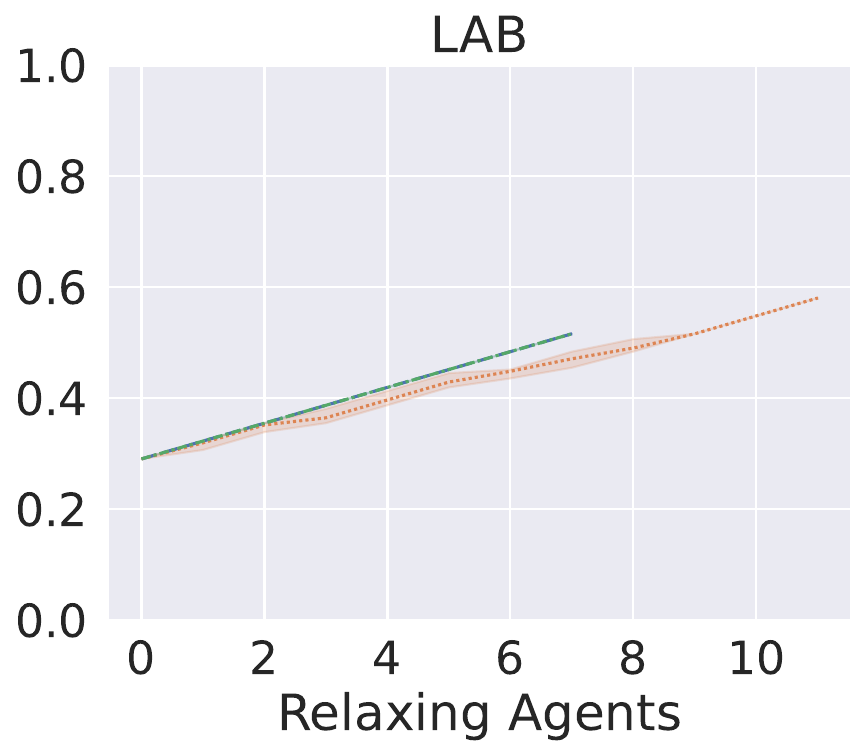}
\includegraphics[width=0.42\columnwidth]{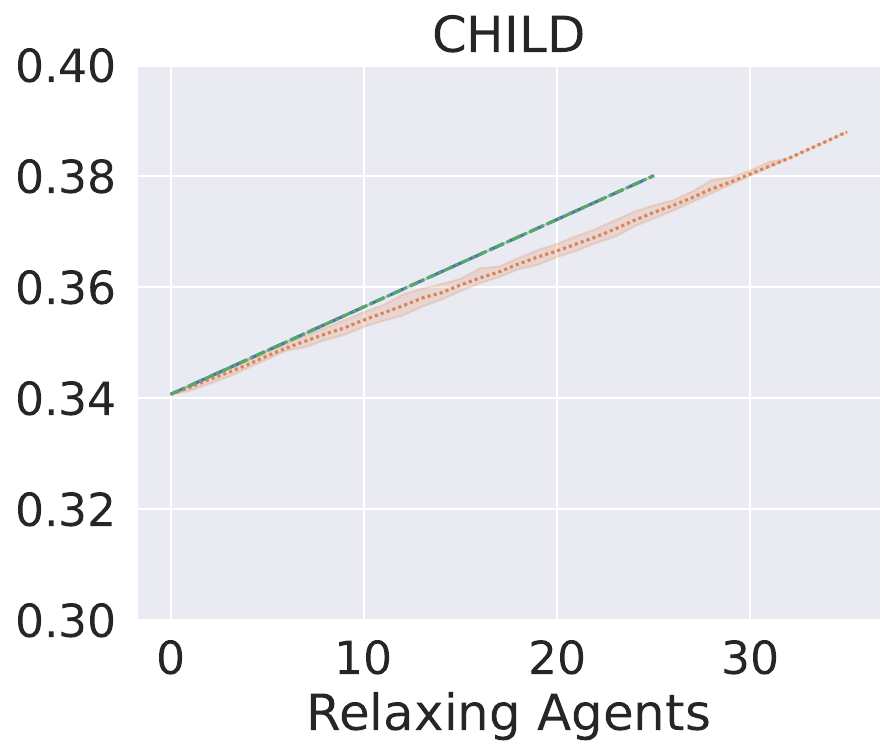}

\caption{Fraction of matched agents vs. number of relaxing agents for SNH-SB, SNH-WB, and WNH-WB guarantees. The scale and limits of the $y$ axis are different for the different datasets.
\label{fig:relaxing_agents}}
\end{figure}

\paragraph{Many-to-one and one-to-many.}

For the COURSE dataset, we allow some courses to have multiple rooms. Such a scenario may happen in exams where there needs to be some gap between students for preventing dishonesty. The CHILD dataset is used when some of the activities allow multiple children for fixing the balance between the children and the avaialabe activities. 
For the CHILD dataset, we chose 1-3 instances at random for each activity. For the COURSE dataset, we again chose 1-3 instances at random for each course.  Similar to Figure~\ref{fig:relaxing_agents}, In Figures~\ref{fig:many} (left) we can see that SNH-SB performs better with low
and moderate compliance rates, while WNH-WB performs better when more agents comply. The bumps in the graphs, where the number of relaxing agents is greatest, is due to the fact that the facilitator only addresses so many agents in a few of the replications and we can therefore safely ignore these bumps. 
In Figure~\ref{fig:many} (right) (COURSE), the fraction of matches is lower than in Figure~\ref{fig:msize}. That is since some courses might need more than one classroom so there is stronger competition on the available classrooms. The opposite is true for the CHILD dataset, where multiple children are allowed per activity.

%\textcolor{red}{(One or two observations about this figure
%will be useful. -- Ravi)}
%\ycomment{Ravi: please check}

\begin{figure}[htb]
\centering

\includegraphics[width=0.45\columnwidth]{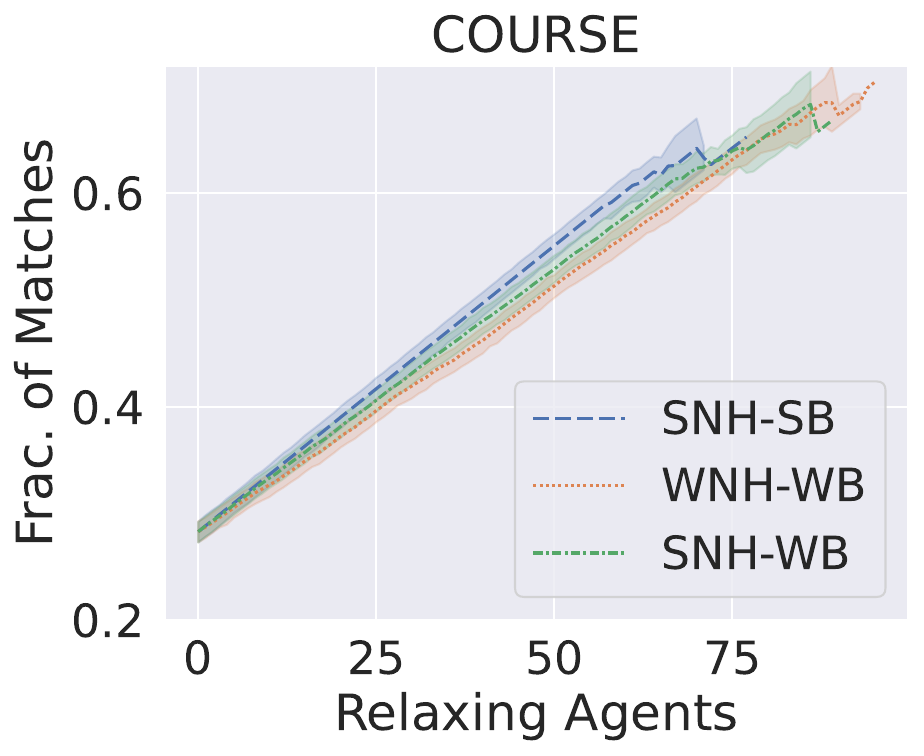}
\includegraphics[width=0.45\columnwidth]{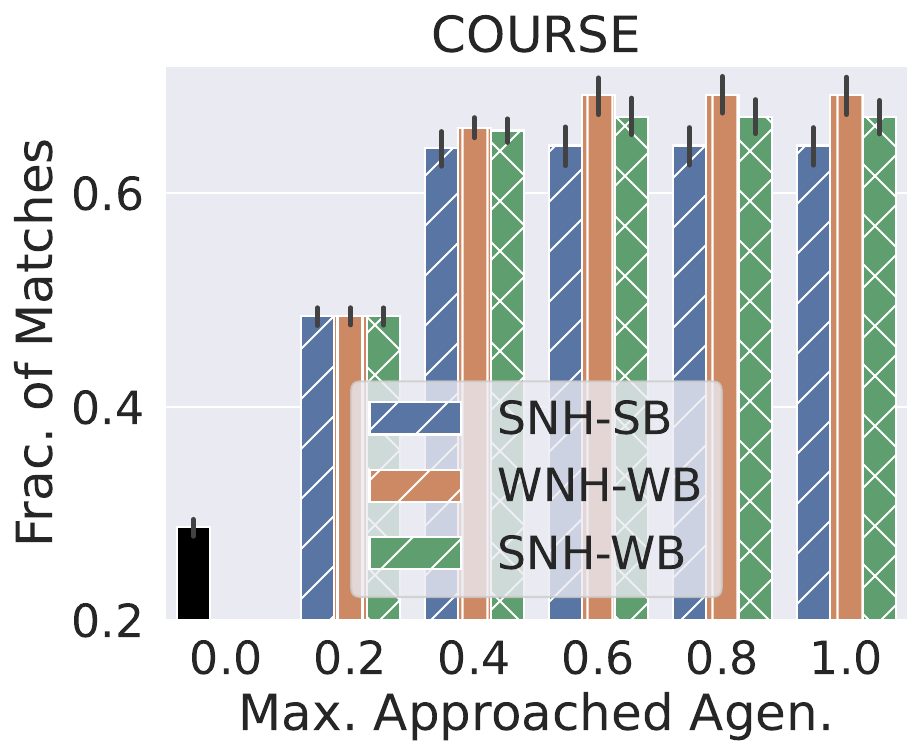}
\includegraphics[width=0.45\columnwidth]{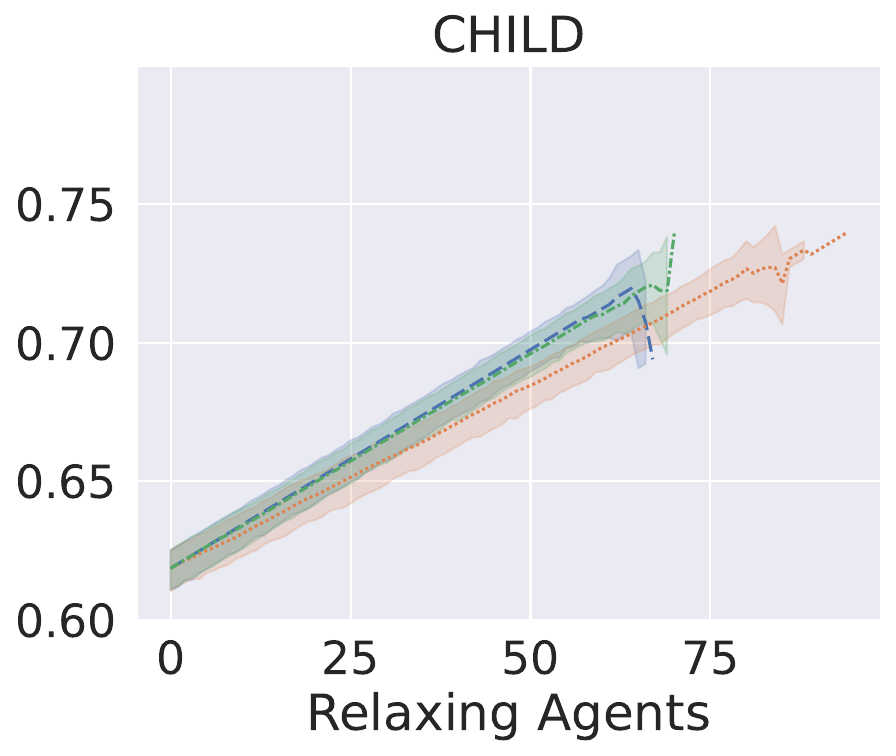}
\includegraphics[width=0.45\columnwidth]{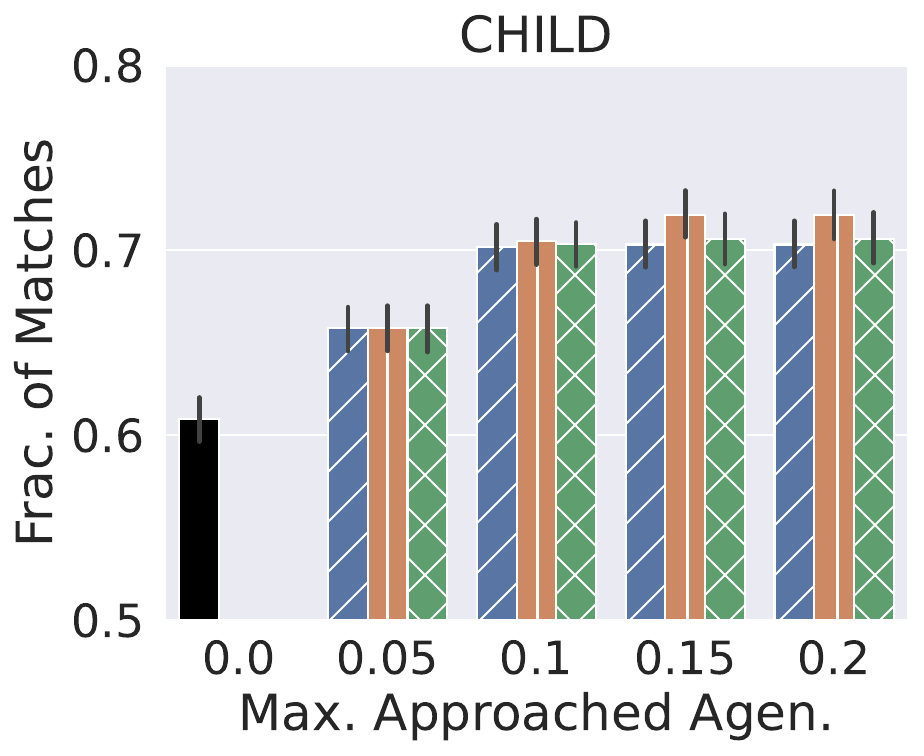}

\caption{Results for many to one (COURSE) and one to many (CHILD)\label{fig:many}}
\end{figure}

%\smallskip
%For space reasons, our experimental results for the children activities, and more results for the courses and classrooms, and the student lab datasets were moved to the Appendix.

%% file: conclusion.tex
\section{Conclusion}
\label{sec:concl}
%\ycomment{I removed the future works as requested by Sarit. Sarit think that having future work is inviting criticism to our paper: why didn't you implemented the ideas in the future works?}
%\textcolor{red}{(Good point. I have restored the original
%version of this section with minor changes. -- Ravi)}

In this work, we define a matching problem where a facilitator advises agents on which constraints to relax. The facilitator's goal is to increase the allocation size while ensuring that no agent is harmed by the suggested relaxations and that relaxing agents benefit from relaxing by securing a guaranteed resource. Additionally, the facilitator is constrained by a bound on the total cost of relaxations that can be suggested.

For three allocation settings, namely one-to-one, many-to-one and one-to-many, we consider three variants of participation guarantees and two possible forms of cost aggregration.
A general approach is presented to tackle the 
one-to-one problems leading to polynomial time algorithms.
An extension of this approach is presented to obtain polynomial time algorithms  for the many-to-one and one-to-many problems as well.
Our experiments demonstrate the usefulness of these algorithms by applying them to three real-world datasets.

Expanding our research by developing other techniques for defining discomfort levels and quantifying them through numerical values is left for future work.

\clearpage

%% file: appendix.tex
\onecolumn
 \appendix
\begin{center}
\fbox{{\Large\textbf{Technical Appendix}}}
\end{center}

\bigskip%\bigskip

\noindent
\textbf{Paper title:}~Facilitating Matches on Allocation Platforms

%\smallskip
%\noindent
%\textbf{Paper id:}~  9712
\medskip

\section{Additional Material for Section~\ref{sec:algorithms}}
\label{sec:appalg}

\subsection{Proof of Theorem~\ref{pro:pcfair}}

We prove the following properties of the set $\hat{E}$ returned by Algorithm~\ref{alg:core}:
\begin{enumerate}
    \item $\hat{E}$ is SNH-SB.
    \item For the bound $\beta$ and for the selected aggregation cost function $g$, $g(\hat{E})\leq \beta$.
    \item $\mu(E\cup \hat{E})$ is maximal among all relaxation sets that fulfil the last two conditions.
    
\end{enumerate}

Throughout, we consider the base graph $G=(X,Y,E)$. We denote $k_{\min}$ as the value of $k$, used in line 14 of Algorithm~ \ref{alg:core} to define $M_{\min}$ and $E_{min}$. 
The values of $M_{\min}$ and $E_{min}$ in line 14 are denoted as $M_{k_{\min}}$ and $E_{k_{min}}$. 
The following lemmata are needed for proving the theorem: 
\begin{statementof}[Lemma~\ref{lem:mwc}]
$|M_{k_{min}}\cap E|=\mu(E)$.
\end{statementof}

\begin{proofof}[Lemma \ref{lem:mwc}]
We first note that since $M_{k_{min}}\cap E$ is a matching of $G$, $|M_{k_{min}}\cap E|\leq \mu(E)$.
It remains to show that $|M_{k_{min}}\cap E|\geq \mu(E)$.
Let $M$ be a maximum matching of $G$ (by definition, $|M|=\mu(E)$), and assume for contradiction that $|M_{k_{min}}\cap E|<\mu(E)$. Now,

\begin{align*}
w(M_{k_{min}})\geq w(M_{k_{min}}\cap E_{k_{min}}) + w(M)=w(M_{k_{min}}\cap E_{k_{min}})+ |M|\cdot(1+|X|)^2 \;\; \\\text{(Since $g(M)=0$ and  $k_{min}\leq |Y|-|M|$; otherwise, there is a smaller $k$ with with $g(\cdot)=0$ })\\
> w(M_{k_{min}}\cap E_{k_{min}}) + (|M|-1)\cdot(1+|X|)^2 + |X|(|X|+1)
\end{align*}

But from the assumption, it follows that:
\begin{align*}
w(M_{k_{min}}) \leq w(M_{k_{min}}\cap E_{k_{min}}) + (|M|-1)\cdot(1+|X|)^2 + |X|(|X|+1) \\ 
\text{(Since the weight of an edge in $E_R$ is less than $(|X|+1)$)}.\\
\end{align*}
We have a contradiction, and the lemma follows.
\end{proofof}

The following lemma complements the results of the previous lemma:

\begin{lemma}
$|M_{k_{min}}\cap E_{k_{min}}|=k_{min}$.
\label{lem:mwc2}
\end{lemma}

\begin{proofof}[lemma~\ref{lem:mwc2}]
Otherwise, there is an unmatched agent $x\in X_{k_{min}}$ with respect to $M_{k_{min}}$. 
Since $|X_{k_{min}}|\leq |Y|$, at least one agent $x'\notin X_{k_{min}}$ is in $X({k_{min}})$. 
Let $y'$ be the resource matched to $x'$ in $M_{k_{min}}$. In this case, from definitions,  $w(M_{k_{min}}\setminus \{(x',y')\}\cup \{(x',y')\})>w(M_{k_{min}})$, 
contradicting the fact that $M_{k_{min}}$ is a maximum weighted matching of $(X\cup X_k,Y, E\cup E_R \cup E_{k_{min}})$.

\end{proofof}

The following claims are also used for proving 
Lemma~\ref{lem:inc}.
(The proof of Lemma~\ref{lem:inc}, which uses these
claims, appears in the main paper.)

\medskip 

\begin{claim}\label{cl:result_1}
 Let $G=(X,Y,E)$ be a bipartite graph and let $x\in \Gamma(E)$. Let $(x,y)\in E_R$ be an edge.
    Then $\mu(E)=\mu(E\cup \{(x,y)\}$.
\end{claim}

\begin{proof}
We first note that since $E\subseteq E\cup \{(x,y)\}$, it follows that $\mu(E)\leq \mu(E\cup \{(x,y)\}$. 
Assume for contradiction that $\mu(E)<\mu(E\cup \{(x,y)\})$. 
Let $M$ be a maximum matching of $G=(X,Y,E)$. We note that $M$ has exactly one edge with $x$ as an endpoint. 
(Since $M$ is a matching, $M$ has at most one edge with $x$ as an endpoint. Also, since $x$ participates in all maximum matchings of $G$, $M$ has at least one such edge.) 

Let $M'$ be a maximum matching of $(X,Y,E\cup \{(x,y)\})$. We note that since $M'$ is a matching, it has at most one edge incident on $x$. If there is no edge incident on $x$ in $M'$, then $M'$ is a matching of $G$ whose size is greater than $\mu(E)$; this contradicts the assumption that $M$ is a maximum matching of $G$. Otherwise, if we remove from $M'$ the edge with $x$ as an endpoint, we get a matching of $G$ with size at least $|M|$ and that matching does not contain $x$. This contradicts the participation of $x$ in all maximum matchings of $G$.
\end{proof}

\begin{claim}\label{cl:result_2}
    Let $G=(X,Y,E)$ be a graph and let $x$ be an agent that does not participate in all maximum matchings of $G$. Let $(x,y)$ be an edge such that $(x,y)\in E_R \notin E$.
    If $x$ participates in all maximum matchings of $(X, Y, E\cup \{(x,y)\})$,  then 
    $\mu(E)+1=\mu(E\cup \{(x,y)\})$.
    %\label{lem:addpincrease}
\end{claim}

\begin{proof}
    We assume that $x$ participates in all maximum matchings of $(\agset, \rset, E\cup \{(x,y)\})$ and prove that $\mu(E)+1=\mu(E\cup \{(x,y)\})$.
    Since $(X,Y,E\cup \{(x,y)\})$ is supergraph of $G$, 
    $\mu(E)\leq \mu(E\cup \{(x,y)\})$. On the other hand, from definition of matching,
    $\mu(E\cup \{(x,y)\})\leq \mu(E)+1$.
    Therefore, either $\mu(E\cup \{(x,y)\}) = \mu(E)$ or $\mu(E\cup \{(x,y)\}) = \mu(E)+1$. Assume for contradiction that $\mu(E\cup \{(x,y)\}) = \mu(E)$. In this case, all maximum matchings of $G$ are also maximum matchings of $(X,Y,E\cup \{(x,y)\})$. From the lemma's definitions, there is a maximum matching of $G$ in which $x$ does not participate. Since it is also a maximum matching of $(X,Y,E\cup \{(x,y)\})$, $x$ does not participate in all maximum matchings of $(X,Y,E\cup \{(x,y)\})$, a contradiction.
\end{proof}

We now present a result that is common knowledge in the area of matching theory
(see e.g., \cite{Lovasz-Plummer-1986}).
We present its proof only for completeness.
In stating this result, we need the following definitions.
A \textbf{free node} of a bipartite graph with respect to a maximum
matching $M$ is one that does not appear as an end point of any 
edge in $M$.
An \textbf{alternating path} with respect to a matching $M$ 
is a simple path in which the first edge is in $M$ and the edges
alternate between those in $M$ and those not in $M$.
An \textbf{even length alternating path} with respect to a matching $M$ 
is an alternating path in which the number of edges is \emph{even}.

\begin{claim}\label{cl:result_3}
Suppose $G(X,Y,E)$ is bipartite
graph. A node $x \in X$ does not participate in all maximum matchings of $G$ if and only if there is an even-length alternating path between $x$ and a free node with respect to any maximum matching of $G$.
\end{claim}
\begin{proof}
We first consider the ``Only if'' part.
Suppose $x \in X$ participates in at least one but not all maximum matchings of $G$.
 Let $M_1$ and $M_2$ be maximum matchings of $G$ with and without $x$ respectively. Consider the 
 symmetric difference $M_1 \bigoplus M_2$, defined by 
 $M_1 \bigoplus M_2$ = $(M_1 \cup M_2) 
 \setminus (M_1 \cap M_2)$.
 Since $x$ participates only in $M_1$, there is a path in $M_1 \bigoplus M_2$ with $x$ as one endpoint. (This path is not a cycle since the  degree of $x$ is $1$. Also, this path does not contain cycles since the maximum degree of a node in $M_1 \bigoplus M_2$ is $2$.) 
 If this path is of odd length, it is an augmenting path for $M_2$,
 and by Berge's lemma, $M_2$ cannot be a maximum matching.
 Otherwise, it is an even-length path that translates to an even-length alternating path between $x$ and a free node in $M_1$. (From the definitions, the path must end with a free node of $M_1$.) Since $x$ is free in $M_2$, $x$ has an even length alternating path of size 0 to a free node in $M_2$. Since this conclusion holds for all matchings 
 $M_1$ and $M_2$ such that $x$ participates in $M_1$ and $x$ does not 
 participate in  $M_2$, the ``Only if"  part follows.  

We consider the ``If'' part of the lemma.
Given a bipartite graph $G=(\agset, \rset, E)$ and any maximum matching $M$ of $G$, suppose there exists an even length alternating path between a vertex $x$ and a free node with respect to $M$.
If this path is of length 0, then $x$ is a free node with respect to
$M$ and so $x$ doesn't participate in all the maximum matchings of $G$.
So, we can assume that the even length alternating path to a free node
must have two or more edges.
In this case, we can replace any two adjacent edges in the alternating path and construct a maximum matching without $x$.
This completes our proof for the ``If'' part and also that of
the Lemma.
\end{proof}

We need the following lemma for proving that  removing the dummy agents does not harm the gained properties:

\begin{lemma}
    \label{lem:remove}
    Let $G=(X,Y,E)$ be a bipartite graph, let $x \in X$ be an agent that participates in all maximum matchings of $G$ and let $x'$ be any other agent. Agent $x$ participates in all maximum matchings of $(X\setminus \{x'\},Y,E\setminus E_{x'})$.        
\end{lemma}

\begin{proof}
    \textbf{Case 1: $x'$ participates in all maximum matchings of $G$.}
    Removing $x'$ from the graph reduces the maximum matching size by 1. Returning $x'$ and its adjucent edges can add either $x'$ or $x$ to the maximum matching but not both (e.g., when expanding via an augmenting path).

    \textbf{Case 2: $x'$ does not participate in all maximum matchings of $G$.}
    In this case, any maximum matching of $(X\setminus \{x'\},Y,E\setminus E_{x'})$ is also a maximum matching of $G$ and therefore $x$ participates in all of them. 
\end{proof}

The following definitions are needed for proving the lemma below:

We define a precedence order~ $\preceq$~ on subsets of $E_R$.
Let $\hat{E},\hat{F}\subseteq E_R$ be two relaxation sets. We say that $\hat{E}\preceq \hat{F}$ if and only if $\hat{E}$ appears before $\hat{F}$ in the order.
Let $u:E_R\rightarrow \mathbb{R}$. We say that  $u$ \emph{represents} $\preceq$ if
\begin{align*}
\hat{F} \preceq \hat{E} \iff \sum_{e\in \hat{F}}u(e) \leq  \sum_{e\in \hat{E}} u(e). 
\end{align*}

We are now ready to prove Lemma~\ref{lem:mkpc} which is essential for proving the correctness of Theorem \ref{pro:pcfair}.

\medskip
\noindent
\textbf{Statement of Lemma~\ref{lem:mkpc}:}

\noindent
The set, $\hat{E}$ returned by  Algorithm~\ref{alg:core} has the following properties:
\begin{enumerate}[label={(\arabic*)}]
    \item It provides the strong benefit for relaxers.
    \item It provides the strong no harm guarantee.
    \item For any SNH-SB set  $\hat{F}\subseteq E_R$,  $|M_{k_{min}}|\geq \mu(E\cup E_{k_{min}}\cup \hat{F})$.
    \item For any SNH-SB set $\hat{F}\subseteq E_R$ and for any aggregation function $g$, if $|M_{k_{min}}|=\mu(E\cup E_{k_{min}} \cup \hat{F} )$, then $\hat{E} \preceq \hat{F}$.        
\end{enumerate}

Set $w^*=|X|+1$ and let $\hat{E}=M_{k_{min}}\setminus E$ 
be the relaxation returned by the algorithm.

\begin{proofof}[lemmas ~\ref{lem:mkpc}(1) and ~\ref{lem:mkpc}(2)]
Consider $\hat{F}\subseteq \hat{E}$. Let$\bar{M}=M_{k_{min}}\cap E$.  
Without loss of generality, assume that $\Gamma(E) \cup \hat{F}\neq \emptyset$.

Recall that by Algorithm \ref{alg:core}, $M_{k_{min}}$ is a matching in $(X,Y,E\cup E_{k_{min}} \cup \hat{E})$ with $M_{k_{min}}=\bar{M} \cup \hat{E}$. So, also $\bar{M}\cup\hat{F}$ is a matching in $(X,Y,E\cup E_{k_{min}}\cup \hat{F})$.  
Consider a maximum matching $\bar{M}_{\hat{F}}$ in $(X,Y,E\cup E_{k_{min}}\cup\hat{F})$. 
Then, 
\begin{align}
    |\bar{M}_{\hat{F}}|\geq |\bar{M}\cup \hat{F}| = |\bar{M}|+|\hat{F}| = \mu(E) + k_{min} + |\hat{F}| \label{eq:MAF}
\end{align}
where the last equality is by Lemmas \ref{lem:mwc} and \ref{lem:mwc2}. 
On the other hand,
\begin{align}
    |\bar{M}_{\hat{F}}|= & 
    ~k_{min}+ | (\bar{M}_{\hat{F}}\cap E) \cup (\bar{M}_{\hat{F}}\cap \hat{F})| \leq
    k_{min} + |\bar{M}_{\hat{F}}\cap E| + |\bar{M}_{\hat{F}}\cap \hat{F}| \nonumber\\
    ~\leq~ & k_{min}+ \mu(E) + |\hat{F}|. 
    \label{eq:MAF2}
\end{align}
So, together with \eqref{eq:MAF}, $|\bar{M}_{\hat{F}}| =k_{min} + \mu(E)+|\hat{F}|$, and the weak inequality of \eqref{eq:MAF2} must be an equality.  So, (i) $|\bar{M}_{\hat{F}}\cap E|=\mu(E)$, and (ii) $|\bar{M}_{\hat{F}}\cap \hat{F}|=|\hat{F}|$.  So, by (i) $\bar{M}_{\hat{F}}\cap E$ is a maximum matching of $(X,Y,E)$, so $\Gamma(E)\subseteq \bar{M}_{\hat{F}}$, and by (ii) $\hat{F}\subseteq \bar{M}_{\hat{F}}$. 
Since it is correct for any $\hat{F}$, $\hat{E}$ has the no-harm and benefit for relaxers SNH-SB guarantees with respect to $(X\cup X_{k_{min}},Y,E\cup E_{k_{min}}\cup \hat{E})$.

By construction of $E_{k_{min}}$, all agents $X_{k_{min}}$ participate in all maximum matchings of $(X\cup X_{k_{min}},Y,E\cup E_{k_{min}}\cup \hat{F})$. 
Therefore, applying lemma~\ref{lem:remove} multiple times and removing the 
$X_{k_{min}}$ agents and the $E_{k_{min}}$ edges, we get that $\hat{E}$ has the SNH-SB guarantees also with respect to $(X,Y,E\cup\hat{E})$ as requested.

\end{proofof}

\begin{proofof}[lemma ~\ref{lem:mkpc}(3)]
For purposes of contradiction, suppose that there is another SNH-SB set $\hat{F}$, for which the maximum matching of $(X\cup X_{k_{min}},Y,E \cup \hat{F} \cup E_{k_{min}})$ is of greater size. 

We set $\hat{F}'\subseteq \hat{F}$ as an inclusion minimal subset of $\hat{F}$ for which $\mu(\hat{F}')=\mu(\hat{F})$.

In particular, since $\hat{F}'$ is inclusion minimal,  $\mu(E\cup E_{k_{min}}\cup{\hat{F}})>\mu(E\cup E_{k_{min}}\cup{\hat{F}}\setminus \{e\}))$.
So, the conditions of Lemma \ref{lem:inc} hold for $(E\cup E_{k_{min}})$ and $\hat{F}'$. 
Let $\bar{M}_{\hat{F}'}$ be a maximum matching of $(X,Y,E\cup E_{k_{min}}\cup{\hat{F}})$. 
So, since $M_{k_{min}}$ is a matching of $(X\cup X_{k_{min}},Y,E\cup E_{k_{min}} \cup \hat{E})$, $|\bar{M}_{\hat{F}'}|> |M_{k_{min}}|$.
So,
\begin{align*}
    |\bar{M}_{\hat{F}'}\cap \hat{F}'| ~=~ & |\bar{M}_{\hat{F}'} \setminus E | = |\bar{M}_{\hat{F}'}| - |\bar{M}_{\hat{F}'}\cap E| \\
    ~=~ &  |\bar{M}_{\hat{F}'}| - \mu(E) \;\; \text{by Lemma~\ref{lem:inc}} \\
    ~>~ & |M_{k_{min}}| - \mu(E) \;\; \text{by assumption} \\
    ~=~ & |M_{k_{min}}| - |M_{k_{min}}\cap E| \;\; \text{by Lemma~\ref{lem:mwc}} \\
    ~=~ & |M_{k_{min}} \cap \hat{E}| +|M_{k_{min}} \cap E_{k_{min}}|
\end{align*}

%So,
\begin{align}
\mbox{So,}\;\;\;    
w(M_{k_{min}}) ~<~ & w(M_{k_{min}}\cap    
    E_{k_{min}}) + |M_{k_{min}}\cap E|(w^*)^2 +    |M_{k_{min}}\cap \hat{E}|w^* \nonumber \\
    ~=~ & w(M_{k_{min}}\cap    
    E_{k_{min}}) + \mu(E)(w^*)^2+ |M_{k_{min}}\cap \hat{E}|w^* \nonumber\\
    ~\leq~& w(M_{k_{min}}\cap    
    E_{k_{min}}) + \mu(E)(w^*)^2+ (|\bar{M}_{\hat{F}}\cap \hat{F}|-1)w^*\nonumber \\
    ~<~  & w(M_{k_{min}}\cap    
    E_{k_{min}}) + \mu(E)(w^*)^2+ |\bar{M}_{\hat{F}'}\cap \hat{F}|w^*-|X|\nonumber \\
    ~<~  & w(M_{k_{min}}\cap    
    E_{k_{min}}) + \mu(E)(w^*)^2+ \sum\nolimits_{e\in \bar{M}_{\hat{F}'}\cap \hat{F}'}w(e) \label{eq:sum}\\
    ~=~ & w(M_{k_{min}}\cap    
    E_{k_{min}}) + w(\bar{M}_{\hat{F}'}\cap E)+w(\bar{M}_{\hat{F}'}\cap \hat{F}') ~=~ w(\bar{M}_{\hat{F}'}) \nonumber \label{eq:ef}\\
\end{align}
(where \eqref{eq:sum} holds since $|\hat{F}'|\leq |X|$ and $w(e)\geq w^*-1$ for all $e$ and \eqref{eq:ef} is from the minimality of $\hat{F}'$).
This contradicts the maximality of $w(M_{k_{min}})$.
\end{proofof}

\begin{proofof}[lemma ~\ref{lem:mkpc}(4)]
Let $\hat{F}$ be a SNH-SB relaxation with maximum allocation when adding $k_{min}$ dummy agents and for which $g(\hat{F})\leq \beta$.

First note that since $\hat{F}$ is SNH-SB so is any subset of $\hat{F}$. So, for any $e\in \hat{F}$, $\mu(E\cup{\hat{F}})>\mu(E\cup{\hat{F}\setminus \{e\}})$, or else  $\hat{F}\setminus \{e\}$ is a SNH-SB set with allocation of the same size that preceding $\hat{F}$ in $\preceq$ (Since $u(\cdot)$ is strictly positive).

For $e\in E_R$, set $u'(e)=\frac{u(e)}{\max_{e'\in E_R}\{ w(e')\}}$.  So, $u'$ is also a representation of $\preceq$, and
\begin{align*}
        w(e) = \begin{cases}
            (w^*)^2  & e\in E \\
            w^*-u'(e) & e\in E_R
        \end{cases}
\end{align*}
We have already established that
\begin{align*}
\mu(E\cup E_{k_{min}} \cup \hat{E}) &=\mu(E)+ k_{min}+ |\hat{E}|.   
\end{align*}
And by Lemma~\ref{lem:inc}
\begin{align*}
\mu(E\cup E_{k_{min}}\cup \hat{F})=\mu(E)+ k_{min}+ |\hat{F}|.  
\end{align*}
So, since both $\hat{E}$ and $\hat{F}$ have allocation of the same size, $|\hat{E}|=|\hat{F}|$.
By Lemma~\ref{lem:mwc}
(for the first equality) and Lemma~\ref{lem:inc} (for the second) we have:
\begin{align*}
w(M_{k_{min}})& ~=~ w(M_{k_{min}}\cap    
    E_{k_{min}})+ \sum_{e\in \hat{E}}(w^*-u'(e))+\mu(E)\cdot (w^*)^2\\
~=~ &w(M_{k_{min}}\cap    
    E_{k_{min}}) -u'(\hat{E})+|\hat{E}|w^*+\mu(E)\cdot (w^*)^2\\
w(\bar{M}_{\hat{F}})& ~=~ w(M_{k_{min}}\cap    
    E_{k_{min}})+ \sum_{e\in \hat{F}}(w^*-u'(e))+\mu(E)\cdot (w^*)^2\\
~=~ &w(M_{k_{min}}\cap    
    E_{k_{min}}) -u'(\hat{F})+|\hat{F}|w^*+\mu(E)\cdot (w^*)^2
\end{align*}
By construction $w(M_{k_{min}})\geq w(\bar{M}_{\hat{F}})$.  So, since $|\hat{E}|=|\hat{F}|$, $u'(\hat{E})\leq u'(\hat{F})$.  So, $\hat{E}\preceq\hat{F}$.
\end{proofof}

We now provide a lemma for showing that Lemma~\ref{lem:mkpc} is still valid on $(X,Y,E\cup \hat{E})$. 

%The following lemma is used for proving %theorem~\ref{pro:pcfair}:

\begin{lemma}
\label{lem:gequ}
Given two sets, $\hat{E}$ and $\hat{F}$,
\begin{enumerate}
\item If $u_1(\hat{E})<u_1(\hat{F})$ then \mbox{\bf $\Si_\rho(\hat{E})$} $<$ \mbox{\bf $\Si_\rho(\hat{F})$}.
\item If $u_2(\hat{E})<u_2(\hat{F})$ then \mbox{\bf $\Ut(\hat{E})$} $<$ \mbox{\bf $\Ut(\hat{F})$},~ and
\end{enumerate}

\end{lemma}

\begin{proofof}[lemma~\ref{lem:gequ}]
   The lemma follows directly from definitions. 
\end{proofof}

\begin{proofof}[Theorem ~\ref{pro:pcfair}]

We first prove that $\hat{E}$ is SNH-SB. 
By Lemmas \ref{lem:mkpc}(1) and (2), $\hat{E}$ is SNH-SB in $(X\cup X-{k_{min}}, Y, E \cup E_{K_{min}} \cup \hat{E})$.
Applying Lemma~\ref{lem:remove} multiple times, we get that $\Gamma((E\cup E_k \cup \hat{E}))\cap (E \cup \hat{E}) \subseteq \Gamma(E \cup \hat{E})$.

 We conclude that the set $\hat{E}$ is has the SNH-SB guarantee for both no-harm and benefit to relaxers.
As we showed in lemma~\ref{lem:mkpc}(3), the solution maximizes the size of the allocation given $k_{min}$. From Lemmas~\ref{lem:mkpc}(4) and ~\ref{lem:gequ}, it follows that the returned allocation has the maximal size regardless of the chosen value of $k$. 
(If there was a smaller $k$ with feasible allocation and within the bound, the algorithm would stop earlier).
Therefore, Theorem~\ref{pro:pcfair} is proved.
\end{proofof}

\subsection{Weak no harm weak benefit guarantee (WNH-WB)}
\label{sec:fc}
For optimizing the allocation size given the weaker WNH-WB guarantees, we need to modify the weight function which is now defined as follows:

\begin{align*}
        w(e) = \begin{cases}
            |X|+1, & e\in E, \\
            |X|+1-\frac{u(e)}{\max_{e'\in E_R}\{ u(e')\}}, & e\in E_R\,.
        \end{cases}
\end{align*}

For WNH-WB (appendix), the weights of the original edges are only slightly greater than the weight of the relaxed edges; this way we prioritize the size of the matching over the strength of the guarantees.

We now prove 
%in Section~\ref{sec:appalg} of the appendix 
that Algorithm \ref{alg:core} with the above weight
function $w$ solves the WNH-WB Optimization Problem.
The following theorem summarizes our  results for the WNH-WB optimization problem.

\begin{theorem}
Algorithm~\ref{alg:core} with the weight function $w(\cdot)$ solves the
WNH-WB optimization problem with an aggregate cost function $g$ (see Theorem~\ref{pro:pcfair}) and a bound $\beta$.
\label{pro:mainfc}
\end{theorem}

Our proof of Theorem~\ref{pro:mainfc} uses the following lemma.

\begin{lemma}
\label{lemma:O-connected}
    For any $B=(X,Y,F)$, $\Gamma(F)$ and $X\setminus \Gamma(F)$ reside in separate connected components.   
\end{lemma}
\begin{proofof}[Lemma~\ref{lemma:O-connected}]
If the lemma does not hold, then there exist 
$x\in \Gamma(F),x'\not\in \Gamma(F), y\in Y$, with $(x,y),(x',y)\in F$. 
Let $M$ be a maximum matching in which $x'$ does not participate. In particular, $(x',y)\not\in M$.  Define $M'=M\setminus \{(x,y)\} \cup \{ (x',y)\}$. Then, $|M'|\geq |M|$.  But $M$ is maximum.  So $|M'|=|M|$ and  thus $(x,y)\in M$. So $x$ does not participate in the maximum matching $M'$. This is a contradiction and
the lemma follows.

\end{proofof}

Recall the notations:  $X_{k_{min}}$ is the set of dummy agents added by the algorithm for $k_{min}$. $M_{k_{min}}, E_{k_{min}}$ are the computed maximum weighted matching and edge set respectively.
For a set $F$ of edges, $X(F)$ are the nodes from $X$ with edges in $F$, and for a graph $B=(X,Y,F)$, $\Gamma(F)$ is the set of $x\in X$ that participate in all maximum matchings of $B$.
$M_P$ is any maximum matching of $\hat{G}$ returned by the platform. 
For a function $f:D\rightarrow \mathbb{R}$ and set $D'\subseteq D$, let $f(D')=\sum_{d\in D'}f(d)$.

\smallskip 

We are now ready to present the following lemma which is essential for proving Theorem~\ref{pro:mainfc}:

\begin{lemma}
The set, $\hat{E}$ returned by  Algorithm~\ref{alg:core} has the following properties:
\begin{enumerate}[label={(\arabic*)}]
    \item It provides the weak benefit for relaxers.
    \item It provides the weak no harm participation guarantee.
    \item $M_{k_{min}} ~\geq~ \mu(E\cup E_{k_{min}} \cup E_R )$. 
    \item For all $\hat{F}\subseteq E_R$, if $|M_{k_{min}}|=\mu(E\cup E_{k_{min}} \cup \hat{F} )$ then $\hat{E} \preceq \hat{F}$.        
\end{enumerate}
\label{lem:mk}
\end{lemma}

\begin{proofof}[Lemma \ref{lem:mk}(1)]
We consider two cases.
In the first case, there is an agent $x_{k_{min}}\in X_{k_{min}}$ such $x_{k_{min}}$ is not adjacent to any edge in $M_P$. Therefore, the set $\hat{E}$ is considered also for $k-1$ and by lemma~\ref{lem:gequ}, $k_{min}$ is not the smallest number, contradicting lines 11 through 13 of Algorithm~\ref{alg:core}.  
 
In the second case, all agents of $X_{k_{min}}$ are matched in $M_P$.
It suffices to prove that $\hat{E}\subseteq M_P$.  Set $F=\hat{E}\setminus M_P$.  Note that $M_P\subseteq E\cup (\hat{E}\setminus F)$. 
Set $j=|M_P|=|M_k|$ (where the equality was established above) and $w^*=|X|+1$.
Then:
\begin{align}
w(M_{k_{min}}) ~=~ & w(E_{k_{min}}) + w(\hat{E})+(j-|\hat{E}|)w^* \nonumber\\
~\leq~  & w(E_k) + w(\hat{E})-w(F)+|F|w^* +(j-|\hat{E}|)w^* \label{eq:strict}\\
& \;\;\;\;\;\;\text{(since $\max_{e\in E_R} \{ w(e)<w^*$)}\nonumber \\
~=~  & w(E_{k_{min}}) + w(\hat{E}\setminus F)+(k-|\hat{E}\setminus F|)w^* \nonumber\\
& \;\;\;\;\;\;\text{\{ since $F\subseteq \hat{E}$\}}\nonumber\\
~=~ &w(M_P), \nonumber 
\end{align}
and the inequality in \eqref{eq:strict} is strict if $F\neq \emptyset$.  So, since $w(M_{k_{min}})$ is maximal, it must be that $F=\emptyset$.
\end{proofof}

\begin{proofof}[Lemma \ref{lem:mk}(2)]

For purposes of contradiction, let $x_0\in \Gamma(E)$ be an agent that does not participate in $M_P$. Let $M$ be a maximum matching of $(X,Y,E)$. By definition, all agents of $\Gamma(E)$ participate in $M$. Define a path $Q=(x_0,y_0,x_1,y_1,\ldots, x_{\ell},y_{\ell})$, with $x_0,\ldots, x_{\ell}\in \Gamma(E)$, inductively as follows:  \newline

\noindent
$\bullet$ $y_i$: is the $y'$ such that $(x_i,y')\in M$ (this exists since $x_i\in \Gamma(E)$, as we shall show).

\noindent
$\bullet$ $x_{i+1}$: is the $x'\in X\cup X_{k_{min}}$ such that $(x',y_i)\in M_P\setminus \hat{E}$, if such an $x'$ exists.  If no such $x'$ exists, then the path ends.

\smallskip 

Note that indeed $x_i\in \Gamma(E)$, for all $i=1,\ldots,\ell$.  This holds since, by Lemma \ref{lemma:O-connected} only vertices of $\Gamma(E)$ are reachable from $x_0$ using $E$, and the definitions of both $y_i$ and $x_{i+1}$ use only edges of $E$. 

We consider two cases, both of which result in a contradiction. 
\begin{itemize}
    \item There is no $(x,y_{\ell})\in M_P$: so the path $Q$ is an augmenting path in $M_P$, and by Berge's lemma, $M_P$ cannot be maximum.  
    \item There exists $(x,y_{\ell})\in M_P$:  So, $(x,y_{\ell})\in \hat{E}\cup E_{k_{min}}$.  But then consider the matching 
    \begin{align*}
        M'_P=(M_P\setminus  \{(x,y_{\ell} \} \setminus \{ (x_{i+1},y_i) : i=0,\ell-1 \} ) \\
        \cup \{ (x_i,y_i): i=0,\ldots,\ell\} .
    \end{align*}
This matching is obtained from $M_P$ by removing $\ell+1$ edges and adding the same number of edges. So, it is of the same size. But, $M'_P$ is obtained with relaxing one less edge. Furthermore, $M'_P\cap \hat{E}\subseteq M_P\cap \hat{E}$. This
contradicts Lemma~\ref{lem:mk}(1).
\end{itemize}
\end{proofof}

\begin{proofof}[Lemma \ref{lem:mk}(3)]
Set $G=(X,Y,E), \hat{G}=(X\cup X_{k_{min}},Y,E\cup E_{k_{min}} \cup \hat{E})$. 
By construction $w(e)\geq |X|$ for all $e\in (E \cup E_R \cup E_{k_{min}})$.  
Note that since $u$ is strictly positive, $\hat{F}\prec \hat{E}$, 
whenever $\hat{F}\subsetneq \hat{E}$.

We prove that $M_{k_{min}}$ is a maximum matching of $(X\cup X_{k_{min}},Y,E\cup E_R\cup E_{k_{min}})$. 
To get a contradiction, suppose $M$ is a matching of $\hat{G}$ with $|M|>|M_{k_{min}}|$. Then,
\begin{align*}
    w(M)\geq |M|\cdot |X| &\geq 
    (|M_{k_{min}}|+1)|X| \\  &> |M_{k_{min}}|(|X|+1) \geq w(M_{k_{min}})
\end{align*}
(where the inequalities hold since $|M_{k_{min}}|<|M|\leq |X|$).  This contradicts the maximality of $w(M_{k_{min}})$.  So, $M_{k_{min}}$ is a maximum matching of $(X\cup X_{k_{min}},Y,E\cup E_R\cup E_{k_{min}})$. It is also a possible matching of $\hat{G}$.  So, since $M_P$ is a maximum matching of $\hat{G}$, also $|M_P|\geq |M_{k_{min}}|$.  So, $M_P$ is a maximum matching of $(X,Y,E\cup E_R \cup E_{k_{min}})$.

We conclude that $M_{k_{min}}$ is a maximum matching of $(X\cup X_{k_{min}},Y,E\cup E_R \cup E_{k_{min}})$. Therefore, it is the maximum sized allocation when adding $k_{min}$ agents. 

\end{proofof}

\begin{proofof}[lemma~\ref{lem:mk}(4)]
Set $j = |M_{k_{min}}|$. Consider any other set $\hat{F}\subseteq E_R$. Let $M$ be a maximum matching of $(X,Y,E\cup E_{k_{min}}\cup \hat{F})$. 
From lemma~\ref{lem:mk}(1) it follows that $|M|=|M_{k_{min}}|$.

If $\hat{F}\not \subseteq M$, then $M\cap \hat{F}$  precedes $\hat{F}$ in $\preceq$ (since $M\cap \hat{F} \subsetneq \hat{F}$). So, we may assume that $\hat{F}\subseteq M$.

For $e\in E_R$, set $u'(e)=\frac{u(e)}{\max_{e'\in E_R}\{ w(e')\}}$.  So, $u'$ also represents $\preceq$, and
\begin{align*}
        w(e) = \begin{cases}
            w^*  & e\in E \\
            w^*-u'(e) & e\in E_R
        \end{cases}
\end{align*}
Then:
\begin{align*}
w(M_{k_{min}}\setminus E_{k_{min}})&=\sum_{e\in \hat{E}}w(e)+(k-|\hat{E}|)w^*\\
&=\sum_{e\in \hat{E}}(w^*-u'(e))+(j-|\hat{E}|)w^* =
jw^*-u'(\hat{E})
\end{align*}
Similarly,
\begin{align*}
w(M\setminus E_{k_{min}})&= jw^*-u'(\hat{F}).       
\end{align*}

So, if $|M_{k_{min}}\cap E_{k_{min}}|\leq |M\cap E_{k_{min}}|$, then the lemma is proved. 
Otherwise, since $g(\hat{F})\leq \beta$, the set $\hat{F}$ was considered for $k_{min}-1$.
From Lemma~\ref{lem:gequ}, it follows that Algorithm~\ref{alg:core} should have stopped at $k_{min}-1$ or before and we have a contradiction.
\end{proofof}

Applying Lemma~\ref{lem:remove} multiple times on $(X\cup X_{k_{min}},Y,E\cup E_{k_{min}}\cup \hat{E}$), we get that $\Gamma((E\cup E_k \cup \hat{E}))\cap (E \cup \hat{E}) \subseteq \Gamma(E \cup \hat{E})$. 
%Using lemma 

We conclude that the set $\hat{E}$ is has the WNH-WB guarantee for both no-harm and benefit to relaxers.
As we showed in lemma~\ref{lem:mk}(3),  the solution maximizes the size of the allocation given $k_{min}$. From Lemmas~\ref{lem:mk}(4) and ~\ref{lem:gequ} it follows that the returned allocation has the maximal size regardless of the chosen value of $k$. 
(If there was a smaller $k$ with feasible allocation and within the bound, the algorithm would stop earlier).
Therefore, Theorem~\ref{pro:mainfc} is proved.

\subsection{Proof of Theorem~\ref{pro:mc}}
We note that most of the correctness proof is very similar to the correctness proof of Theorem~\ref{pro:mainfc}.
Therefore, we focus here on the differences.
Actually, it remains to explain (and prove) that approaching agents from $\Gamma(E)$ results in violating the stronger no-harm guarantee. That is the goal of the following lemma:

\begin{lemma}
Let $\hat{E}$ be the set of edges returned 
by Algorithm~\ref{alg:core} and let $x\in X$ be an agent that participates in all maximum matchings of $(X,Y,E)$.
If there is an edge $\{x,y\}\in \hat{E}$, then $x$ does not participate in all maximum matchings of $(X,Y,E\cup\hat{E}\setminus \{x,y\})$.
\label{lem:mixed}
\end{lemma}

\begin{proofof}[lemma~\ref{lem:mixed}]
Assume for purposes of contradiction that $x$ participates in all maximum matchings of $(X,Y,E\cup\hat{E}\setminus \{x,y\})$.
Consider the set $\hat{E}' = \hat{E}\setminus \{x,y\}$.
We note that from definition of the aggregate functions, $g(\hat{E}')\leq g(\hat{E})$.
Therefore, $\hat{E}'$ was considered in iteration $k_{min}$.
Let $M$ be a maximum weighted matching of $(X\cup X_{k_{min}},Y,E\cup E_{k_{min}} \cup (\hat{E}\setminus \{x,y\})$.
 From Lemma~\ref{lem:mk}.1 and from definition of weight function $w$, we have $w(M) > w(M_{k_{min}})$ contradicting the assumption.
\end{proofof}

The definition of the weight function indeed avoids edges of $\Gamma(E)$ and allows all other edges. The rest of the proof is similar to that of Theorem~\ref{pro:mainfc}.

\section{Additional material for Section~\ref{sec:experiments}}
\label{sec:appexp}

\subsection{Implementation Details}

Experiments were carried out on both a Windows laptop and a Linux server, using Python 3.9 and the NetworkX 2.6 library. Each experiment was repeated 10 times, and in cases involving randomness, the default seeds from Python and the operating system were applied. Testing with different seeds yielded consistent results. In Figures 3-5, error bars indicate confidence intervals of 0.95.